\documentclass[bj]{imsart}

\pdfoutput=1
\usepackage{algorithm}
\usepackage{algorithmicx}
\usepackage{algpseudocode}
\usepackage{color}
\usepackage{verbatim}
\usepackage{amsfonts}
\usepackage{natbib}
\usepackage{hyperref}
\usepackage{amsthm}
\usepackage{amsmath}
\usepackage{amssymb}
\usepackage{titlesec}
\usepackage{subcaption}
\usepackage{graphicx}
\usepackage{clipboard}

\newcommand{\intd}{\mbox{d}}
\newcommand{\xvec}{x}

\newcommand{\zvec}{z}
\DeclareMathOperator*{\argmin}{arg\,min}

\algdef{SE}[DOWHILE]{Do}{doWhile}{\algorithmicdo}[1]{\algorithmicwhile\ #1}%
 
\newtheorem{theorem}{Theorem}[section]
\newtheorem{proposition}{Proposition}[section]
\newtheorem{lemma}{Lemma}[section]

\newclipboard{clipboard}

\begin{document}

\begin{frontmatter}

\title{Non-Homogeneous Poisson Process Intensity Modeling and Estimation using Measure Transport}

\begin{aug}
  \author{\fnms{Tin Lok James}  \snm{Ng}},
  \and
  \author{\fnms{Andrew}  \snm{Zammit-Mangion}}%

  \runauthor{Ng, T.L.J. and Zammit-Mangion, A.}

  \affiliation{School of Mathematics and Applied Statistics, University of Wollongong, Australia}

  \address{School of Mathematics and Applied Statistics, University of Wollongong, Australia}


\end{aug}

\date{}
\maketitle

\begin{abstract}
Non-homogeneous Poisson processes are used in a wide range of scientific disciplines, ranging from the environmental sciences to the health sciences. Often, the central object of interest in a point process is the underlying intensity function. Here, we present a general model for the intensity function of a non-homogeneous Poisson process using measure transport. The model is built from a flexible bijective mapping that maps from the underlying intensity function of interest to a simpler reference intensity function. We enforce bijectivity by modeling the map as a composition of multiple bijective maps that have increasing triangular structure, and show that the model exhibits an important  approximation property. Estimation of the flexible mapping is accomplished within an optimization framework, wherein computations are efficiently done using tools originally designed to facilitate deep learning, and a graphics processing unit. Point process simulation and uncertainty quantification are straightforward to do with the proposed model. We demonstrate the potential benefits of our proposed method over conventional approaches to intensity modeling through various simulation studies. We also illustrate the use of our model on a real data set containing the locations of seismic events near Fiji since 1964.
 

\end{abstract}

\begin{keyword}
\kwd{Poisson Point Process}
\kwd{Intensity Estimation}
\kwd{Measure Transport}
\kwd{Deep Neural Network}
\end{keyword}

\end{frontmatter}

\section{Introduction}
A non-homogeneous Poisson process (NHPP) is a Poisson point process that has variable intensity in the domain on which it is defined.
NHPPs are commonly used in a wide range of applications, for example when modeling failures of repairable systems \citep{Lindqvist2006}, earthquake occurrence \citep{Hong1995}, or the evolution of customer purchase behavior \citep{Letham2016}. 
\\\\
A NHPP defined on ${\cal S} \subset \mathbb{R}^{d}$ can be fully characterized through its intensity function $\lambda: {\cal S} \rightarrow [0, \infty)$.  The intensity function is usually of considerable scientific interest, and both parametric and nonparametric methods have been proposed to model it. A parametric approach assumes that the intensity function has a known parametric form, and that the model parameters can be estimated using, for example, likelihood-based methods \citep[e.g.,][]{Zhao1996}. The specified functional form is, however, often too restrictive an assumption in practice. Non-parametric techniques, on the other hand, do not fix the functional form of the intensity function. Methods in this class for modeling the intensity function include ones that are spline-based \citep{Dias2008}, wavelet-based \citep{Kolaczyk1999, Miranda2011}, and kernel-based \citep{Diggle1985}. While non-parametric methods offer greater modeling flexibility, they often do not scale well with the number of observed points or the dimension $d$.
\\\\
Bayesian methods can be adopted for intensity function estimation if one has prior knowledge (e.g., on the function's smoothness) that could be used. This prior knowledge is often incorporated by treating the intensity function as a latent stochastic process; the resulting model is called a doubly-stochastic Poisson process, or Cox process \citep{Moller1998}. One popular variant of the Cox process is the trans-Gaussian Cox process, where a transformation of the intensity function is a Gaussian process (GP). Inference for such models typically requires Markov chain Monte Carlo methods \citep{Adams2009}, which scale poorly with the number of observed points and dimension $d$. Approximate Bayesian methods such as variational inference \citep{Zammit2011, Lloyd2015}, or Laplace approximations \citep{Illian2012}, often impose severe, and sometimes inadequate, restrictions on the functional form of the posterior distributions. 
\\\\
The models for the intensity function discussed above either place assumptions on the intensity function that are overly restrictive, or require computational methods that are inefficient, in the sense that they do not scale well with data size and/or the dimension $d$. Here, we present a new model for the intensity function that overcomes both limitations. The model finds its roots in transportation of probability measure \citep{Marzouk2016}, an approach that has gained popularity recently for its ability to model arbitrary probability density functions. The basic idea of this approach is to construct a ``transport map'' between the complex, unknown, intensity function of interest, and a simpler, known,  reference intensity function. 
\\\\
We use a map that is sufficiently complex for it to approximate arbitrary intensity functions on subsets of $\mathbb{R}^d$, and one that is easy to fit to observational data. Specifically, we construct a transport map through compositions of several simple \emph{ increasing triangular maps} \citep{Marzouk2016}, in a procedure sometimes referred to as map stacking \citep{Papamakarios2017}.  Our model has the ``universal property'' \citep{Hornik1989}, in the sense that a large class of intensity functions can be approximated arbitrarily well using this approach. We estimate the parameters in the map using an optimization framework wherein \Copy{GPUs}{computations are carried out efficiently on graphics processing units using software libraries created to facilitate deep learning}. We also develop a technique to efficiently generate a realization from the fitted point process, and a nonparametric bootstrap approach \citep{Efron1981} to quantify uncertainties on the estimated intensity function via the stack of increasing triangular maps. 
\\\\
The article is organized as follows. Section~\ref{Bg_sec} establishes the notation and the required theoretical background on transportation of probability measures, while Section~\ref{method_sec} presents our proposed method for intensity function modeling and estimation of NHPPs, and also a theorem relating to the universal approximation property of our model. Results from simulation and real-application experiments are given in Section~\ref{experiments}. Section \ref{conc_sec} concludes. Additional technical material is provided in Appendix~\ref{Appendix}.
\\\\



\section{Transportation of Probability Measure}
\label{Bg_sec}

Our methodology for intensity function modeling in Section~\ref{method_sec} is based on measure transport, and techniques that enable it for density estimation. In Section~\ref{mt_sec} we briefly describe measure transport and increasing triangular maps. In Section~\ref{tri_sec} we discuss parameterizations of increasing triangular maps and the one we choose in our approach to modeling the intensity function, while in Section~\ref{comp_sec} we briefly discuss the composition of such maps in a deep learning framework.

\subsection{Measure Transport and Increasing Triangular Maps}
\label{mt_sec}
Consider two probability measures $\mu_0(\cdot)$ and $\mu_1(\cdot)$ defined on ${\cal X}$ and ${\cal Z}$, respectively. A transport map $T: {\cal X} \rightarrow {\cal Z}$ is said to push forward $\mu_0(\cdot)$ to $\mu_1(\cdot)$ (written compactly as $T_{\#\mu_0}(\cdot) = \mu_1(\cdot)$) if and only if
\begin{eqnarray}
\label{transport_map}
 \mu_1(B) = \mu_0(T^{-1}(B)), \quad \mbox{for any Borel subset } B \subset {\cal Z}.
\end{eqnarray}
The inverse $T^{-1}(\cdot)$ is treated in the general set valued sense, that is, $\xvec \in T^{-1}(\zvec)$ if $T(\xvec) = \zvec$. If $T(\cdot)$ is injective, then the relationship in (\ref{transport_map}) can also be expressed as
\begin{equation}
\label{transport_map2}
 \mu_{1}(T(A)) = \mu_{0}(A), \quad \mbox{for any Borel subset } A \subset {\cal X}.
 \end{equation}
A transport map satisfying (\ref{transport_map}) represents a deterministic coupling of the probability measures $\mu_0(\cdot)$ and $\mu_1(\cdot)$. An alternative interpretation of the transport map $T(\cdot)$ is that if $v$ is a random vector distributed according to the measure $\mu_0(\cdot)$, then $T(v)$ is distributed according to $\mu_1(\cdot)$.
\\\\
Suppose ${\cal X}, {\cal Z} \subseteq \mathbb{R}^{d}$, and that both $\mu_0(\cdot)$ and $\mu_1(\cdot)$ are absolutely continuous with respect to the Lebesgue measure on $\mathbb{R}^{d}$, with densities $\intd\mu_0(x) = f_0(x) \intd x$ and $\intd\mu_1(z) = f_1(z) \intd z$, respectively.  Furthermore, assume that the map $T(\cdot)$ is bijective differentiable with a differentiable inverse $T^{-1}(\cdot)$ (i.e., assume that $T(\cdot)$ is a $C^{1}$ diffeomorphism), then \eqref{transport_map2} is equivalent to 
\begin{equation} \label{eq:transport}
f_0(\xvec) = f_1(T(\xvec)) |\mbox{det}(\nabla T(\xvec))|, \quad \xvec \in {\cal X}.
\end{equation}


\noindent The conditions under which the map $T_{\#\mu_0}(\cdot) = \mu_1(\cdot)$ exists are established in \cite{Brenier1991} and \cite{Mccann1995}. Of particular note is that $T(\cdot)$ is guaranteed to exist when both $\mu_0(\cdot)$ and $\mu_1(\cdot)$ are absolutely continuous. There may exist infinitely many transport maps that satisfy (\ref{transport_map}). One particular type of transport map is an \emph{increasing triangular map}, that is,
\begin{eqnarray}
\label{triangular_map}
T(\xvec) = (T^{(1)}(x^{(1)}), T^{(2)}(x^{(1)}, x^{(2)}), \ldots, T^{(d)}(x^{(1)}, \ldots, x^{(d)}))',\quad x \in {\cal X},
\end{eqnarray}
where, for $k = 1,\dots, d,$ one has that $T^{(k)}(x^{(1)}, \ldots, x^{(k)})$ is monotonically increasing in $x^{(k)}$. In particular, the Jacobian matrix of an increasing triangular map, if it exists, is triangular with positive entries on its diagonal. Increasing triangular maps have a deep connection with the \emph{optimal transport problem} \citep{Villani2009} which seeks to choose a transport map such that the \emph{total cost of transportation} is minimized. Due to its connection with optimal transport, and because of its structure that leads to efficient computations, we will be exclusively considering this class of maps in the following sections.

\subsection{Parameterization of Increasing Triangular Maps}
\label{tri_sec}
Various approaches to parameterize an increasing triangular map have been proposed \citep[see, for example, ][]{Germain2015, Dinh2015, Dinh2017}. One class of parameterizations is based on the so-called ``conditional networks'' \citep{Papamakarios2017, Huang2018}. Consider for now a map comprising just one increasing triangular map, which we denote as $T_1(\cdot)$ (we will later consider many of these in composition), and let $x \equiv (x^{(1)}, \ldots, x^{(d)})'$. The increasing triangular map $T_1(\cdot)$ we use has the following form: 
\begin{eqnarray}
\label{cond_network}
 T_1^{(1)}(x^{(1)}) &=& S_1^{(1)}(x^{(1)}; \theta_{11}), \nonumber \\
T_1^{(k)}(x^{(1)}, \ldots, x^{(k)}) &=& S_1^{(k)}( x^{(k)}; \theta_{1k}(x^{(1)}, \ldots, x^{(k-1)}; \vartheta_{1k}) ),  \quad k = 2,\dots,d,
\end{eqnarray}
for $x \in {\cal X}$, where $\theta_{1k}(x^{(1)}, \ldots, x^{(k-1)}; \vartheta_{1k}), k = 2,\dots,d, $ is the $k$th ``conditional network'' that takes $x^{(1)}, \ldots, x^{(k-1)}$ as inputs and is parameterized by $\vartheta_{1k}$, and $S_1^{(k)}(\cdot)$ is generally a very simple univariate function of $x^{(k)}$, but with parameters that depend in a relatively complex manner on $x^{(1)}, \ldots, x^{(k-1)}$ through the network. \Copy{CondNet}{Therefore, a conditional network is just a multivariate mapping that takes input $x^{(1)}, \ldots, x^{(k-1)}$ and returns an output in $\mathbb{R}^{m_k}$, where $m_k$ is the number of parameters that parameterize $S_{1}^{(k)}(\cdot)$.} We hence have that $\theta_{1k}: \mathbb{R}^{k-1} \rightarrow \mathbb{R}^{m_k}$. It is often the case that feedforward neural networks are used as the conditional networks \citep{Fine1999}. For ease of exposition, from now on we will slightly abuse the notation and denote $\theta_{1k}(x^{(1)}, \ldots, x^{(k-1)}; \vartheta_{1k})$ simply as $\theta_{1k}$, thereby omitting the explicit dependence on the inputs and the parameters $\vartheta_{1k}$. 
\\\\
One class of maps using conditional networks is that of masked autoregressive flows \citep{Papamakarios2017}. In this class, $m_k = 2,~k = 1,\dots, d$, and the output of the conditional network $\theta_{1k} \equiv (\theta_{1k}^{(1)}, \theta_{1k}^{(2)})' \in \mathbb{R}^2$ parameterizes  $S_1^{(k)}(\cdot)$ as
\begin{eqnarray}
\label{linear_flow}
S_1^{(k)}( x^{(k)}; \theta_{1k}) =  \theta_{1k}^{(1)} + x^{(k)}\exp(\theta_{1k}^{(2)}),\quad x^{(k)} \in {\cal X}^{(k)}.
\end{eqnarray}
In \eqref{linear_flow}, the univariate function $S_1^{(k)}(\cdot)$ is a linear function of $x^{(k)}$ with location parameter $\theta_{1k}^{(1)}$ and scale parameter $\exp(\theta_{1k}^{(2)})$. Monotonicity of $S_1^{(k)}(\cdot)$, and hence of $T_1^{(k)}(\cdot)$, is ensured since $\exp(\theta_{1k}^{(2)}) > 0$.  Another class of such maps is the class of inverse autoregressive flows, proposed by \cite{Kingma2016}. In this class, $m_k = 2, k = 1,\dots,d,$ and 
\begin{eqnarray}
\label{weighted_average_flow}
S_1^{(k)}( x^{(k)}; \theta_{1k} ) = \sigma(\theta_{1k}^{(2)}) x^{(k)} + (1 - \sigma(\theta_{1k}^{(2)})) \theta_{1k}^{(1)},\quad x^{(k)} \in {\cal X}^{(k)},
\end{eqnarray}
where $\sigma(\cdot)$ is the sigmoid function. In this class of maps, each $S_1^{(k)}(x^{(k)})$ outputs the weighted average of $x^{(k)}$ and $\theta_{1k}^{(1)}$, with the weights given by $\sigma(\theta_{1k}^{(2)})$ and $1 - \sigma(\theta_{1k}^{(2)})$, respectively. Monotonicity of $S_1^{(k)}(\cdot)$, and hence of $T_1^{(k)}(\cdot)$, is ensured since $\sigma(\theta_{1k}^{(2)}) > 0$.
\\\\
Both (\ref{linear_flow}) and (\ref{weighted_average_flow}) are generally too simple for modeling density functions or, in our case, intensity functions. In this work we therefore focus on the class of neural autoregressive flows, proposed by \cite{Huang2018}, which are more flexible. In this class, $m_k = 3M$ for $k = 1,\dots, d,$ and $M \ge 1$, and the $k$-th component of the map has the form
\begin{eqnarray}
\label{uni_flow}
S_1^{(k)}( x^{(k)}; \theta_{1k} ) = \sigma^{-1} \Big( \sum_{i=1}^{M} w_{1ki} \sigma( a_{1ki} x^{(k)} + b_{1ki}) \Big),
\end{eqnarray}
where $\sigma^{-1}(\cdot)$ is the logit function, $a_{1ki} \equiv \exp(\theta_{1k}^{(2i)})$, $b_{1ki} \equiv \theta_{1k}^{(3i)}$, and $w_{1ki} \equiv \exp(\theta_{1k}^{(1i)})$ is subject to the constraint $\sum_{i=1}^{M} w_{1ki} = 1$. As with the other two maps discussed above, monotonicity of $S_1^{(k)}(\cdot)$, and hence of $T_1^{(k)}(\cdot)$, is ensured through this construction. The Jacobian of the neural autoregressive flow can be computed using the chain rule since the gradient of both $\sigma(\cdot)$ and $\sigma^{-1}(\cdot)$ are analytically available; this is important for computational purposes since such formulations can easily be handled using automatic differentiation libraries.\\\\
Each univariate function $S_1^{(k)}(\cdot)$ in the neural autoregressive flow comprises two sets of smooth, nonlinear transforms: $M$ sigmoid functions that map from $\mathbb{R}$ to the unit interval, and one logit function that maps from the unit interval to $\mathbb{R}$. The complexity/flexibility of $S_1^{(k)}(\cdot)$ is largely determined by $M$. Note that the neural autoregressive flow has a very similar form to the conventional feedforward neural network with sigmoid activation functions.
\\\\
It is straightforward to see that each component of the increasing triangular map constructed above is differentiable. Indeed, $S_1^{(k)}( x^{(k)}; \theta_{1k} )$ in \eqref{uni_flow} is obviously differentiable with respect to $x^{(k)}$. Also, $S_1^{(k)}( x^{(k)}; \theta_{1k} )$, when treated as a function of $\theta_{1k}$, is clearly differentiable with respect to $\theta_{1k}$, while the conditional network $\theta_{1k}(x^{(1)}, \ldots, x^{(k-1)}; \vartheta_{1k})$ is also differentiable with respect to the input $x^{(1)}, \ldots, x^{(k-1)}$ if it itself is a feedforward neural network with sigmoid activation functions, which we will assume from now on. Therefore, $T_1^{(k)}(x^{(1)}, \ldots, x^{(k)}) = S_1^{(k)}( x^{(k)}; \theta_{1k}(x^{(1)}, \ldots, x^{(k-1)}; \vartheta_{1k}) ) $ is also differentiable with respect to $x^{(1)}, \ldots, x^{(k-1)}$ for $k = 1,\dots, d$.
\\\\
A natural question to ask is how well an arbitrary density function can be approximated by a density constructed using the neural autoregressive flow. It has been shown that the neural autoregressive flow satisfies the `universal property' for the set of positive continuous probability density functions, in the sense that any target density that satisfies mild smoothness assumptions can be approximated arbitrarily well \citep{Huang2018}. We provide a universal approximation theorem for the case of the process density of a Poisson process  in Section \ref{universal}. 

\subsection{Composition of Increasing Triangular Maps}
\label{comp_sec}
It is well known that a neural network with one hidden layer can be used to approximate any continuous function on a bounded domain  arbitrarily well \citep{Hornik1989, Cybenko1989, Barron1994}. However, the size of a single layer network (in terms of the number of parameters) that may be required to achieve a desired level of function approximation accuracy may be prohibitively large. This is important as, despite the universal property of the neural autoregressive flow, both the conditional network and the univariate function $S_1^{(k)}(\cdot)$ in (\ref{uni_flow}) may need to be made very complex in order to approximate a target density up to a desired level of accuracy. Specifically, $M$, as well as the number of parameters appearing in the conditional networks, $\{\vartheta_{1k}\}$, might be prohibitively large. 
\\\\
Neural networks with many hidden layers, known as deep nets, tend to have faster convergence rates to the target function compared to shallow networks \citep{Eldan2016, Weinan2018}. In our case, layering several relatively parsimonious triangular maps through composition is an attractive way of achieving the required representational ability while avoiding an explosion in the number of parameters. Specifically, we let $T(\cdot) = T_N \circ \cdots \circ T_1 (\cdot)$, where $T_l(\cdot), l  = 1,\dots,N$, are increasing triangular maps of the form given in \eqref{cond_network}, parameterized using neural autoregressive flows.
\\\\
The composition does not break the required bijectivity of $T(\cdot)$, since a bijective function of a bijective function is itself bijective. Computations also remain tractable, since the determinant of the gradient of the composition is simply the product of the  determinants of the individual gradients. Specifically, consider two increasing triangular maps $T_1(\cdot)$ and $T_2(\cdot)$, each constructed using the neural network approach described above. The composition $T_2 \circ T_1 (\cdot)$ is bijective, and its gradient at some $x \in \cal X$ has determinant,
$$ \mbox{det}(\nabla T_2 \circ T_1(x)) = (\mbox{det} (\nabla T_1(x)))(\mbox{det} (\nabla T_2(T_1(x)))).$$
Further, since the maps have a triangular structure, the Jacobian at some point $x$ is a triangular matrix, for which the determinant is easy to compute. The determinant of the composition evaluated at $x$ is hence also easy to compute. 

\section{Intensity Modeling and Estimation via Measure Transport}
\label{method_sec}

Consider a NHPP ${\cal P}$ defined on a bounded domain ${\cal X} \subset \mathbb{R}^{d}$, and let $N(\cdot)$ be the stochastic process characterizing ${\cal P}$, where $N(B)$ is the number of events in $B \subseteq \cal X$. A NHPP ${\cal P}$ defined on ${\cal X}$ is completely characterized by its intensity function $\lambda: {\cal X} \rightarrow [0, \infty)$, such that $N(B) \sim \mbox{Poisson}( \mu_{\lambda}(B) ),$ where $ \mu_{\lambda}(B) = \int_{B} \lambda(x) \intd x $ is the corresponding intensity measure. If $\lambda(\cdot) = \lambda$ is constant, the Poisson process is homogeneous. In this section we present our approach for modeling and estimating $\lambda(\cdot)$ from observational data.

\subsection{The Optimization Problem}
\label{optim_problem}

The density of a Poisson process does not exist with respect to the Lebesgue measure. It is therefore customary to instead consider the density of the NHPP of interest with respect to the density of the unit rate Poisson process, that is, the process with $\lambda(\cdot) = 1$. We denote the resulting density as $f_\lambda(\cdot)$. Let $|B|$ denote the Lebesgue measure of a bounded set $B \subset \mathbb{R}^{d}$, and let $X \equiv \{x_1, \ldots, x_n\},$ where $x_i \in {\cal X}, i = 1,\ldots,n,$ and $n \ge 1,$ be a realization of ${\cal P}$. The density function $f_\lambda(\cdot)$ evaluated at $X$ is given by,
\begin{align}
    f_{\lambda}(X) &= \exp(|{\cal X}| - \mu_{\lambda}({\cal X})) \prod_{x \in X} \lambda(x) \nonumber \\
         &= \exp \Big(- \int_{{\cal X}} (\lambda(x) - 1) \intd x + \sum_{x \in X} \log \lambda(x)  \Big) \label{fdensity}.
\end{align}




\noindent Our objective is to estimate the unknown intensity function $\lambda(\cdot)$ that generates the data $X$. A commonly employed strategy is to estimate $\lambda(\cdot)$ using maximum likelihood. It is well known that maximizing the likelihood is equivalent to minimizing the Kullback--Leibler (KL) divergence between the true density and the estimate. For two probability densities $p(\cdot)$ and $q(\cdot)$, the KL divergence is defined as $ D_{KL}(p||q) = \int p(x) \log (p(x)/ q(x))\intd x$.  We therefore estimate the unknown intensity function $\lambda(\cdot)$ by $\hat\lambda(\cdot)$, as follows,
\begin{eqnarray}
\label{min_kl}
 \hat{\lambda}(\cdot) = \argmin_{\rho(\cdot) \in {\cal A} }\big\{ D_{KL}(f_{\lambda}||f_{\rho}) \big\} ,
\end{eqnarray}
where ${\cal A}$ is some set of intensity functions, and $f_\rho(\cdot)$ is the density of a NHPP with intensity function $\rho(\cdot)$ taken with respect to the density of the unit rate Poisson process. \\\\
To solve the optimization problem defined in \eqref{min_kl}, we first derive the following expression for the KL divergence between two arbitrary densities.
\begin{proposition}
\label{kl_prop}
\Copy{KL}{Consider two Poisson processes ${\cal P}_1$, ${\cal P}_2$ on ${\cal X}$ with intensity functions $\rho_1(\cdot)$ and $\rho_2(\cdot)$, respectively. Denote the corresponding densities with respect to the unit rate Poisson process as $f_{\rho_{1}}(\cdot)$ and $f_{\rho_{2}}(\cdot)$, where the probability measure corresponding to the density $f_{\rho_{1}}$ is absolute continuous with respect to the probability measure corresponding to $f_{\rho_{2}}$. The Kullback-Leibler divergence $D_{KL}(f_{\rho_{1}}||f_{\rho_{2}})$ is: 
$$  D_{KL}(f_{\rho_{1}}||f_{\rho_{2}}) = \int_{{\cal X}} (\rho_{2}(x) - \rho_{1}(x) ) \textup{d} x + \int_{{\cal X}} \rho_{1}(x) \log \frac{ \rho_{1}(x) }{\rho_{2}(x) }\textup{d} x.$$}
provided that the integrals on the right hand side exist.
\end{proposition}
\noindent We give a proof for Proposition~\ref{kl_prop} in  Appendix~\ref{proofKLsec}. 
\\\\
 In order to apply the measure transport approach to intensity function estimation, we first define $\tilde{\rho}(\cdot) = \rho(\cdot) / \mu_{\rho}({\cal X}) $ and $\tilde{\lambda}(x) =  \lambda(x) / \mu_{\lambda}({\cal X}),$ so that $\tilde{\rho}(\cdot)$ and $\tilde{\lambda}(\cdot)$ are valid density functions with respect to Lebesgue measure. In particular, $\tilde{\rho}(\cdot)$ and $\tilde{\lambda}(\cdot)$ are termed \emph{process densities} by \citet{Taddy2010}, which can be modeled separately from the integrated intensities $\mu_{\rho}({\cal X})$ and $\mu_{\lambda}({\cal X})$, respectively. The KL divergence $D_{KL}(f_\lambda || f_\rho)$ can be written in terms of process densities as follows,
\begin{eqnarray}
\label{kl_min2}
 D_{KL}(f_{\lambda}||f_{\rho}) = \mu_{\rho}({\cal X)} - \mu_{\lambda}({\cal X}) \int_{{\cal X}} \tilde{\lambda}(x) \log \tilde{\rho}(x) \intd x - \mu_{\lambda}({\cal X}) \log \mu_{\rho}({\cal X}) + \textrm{const.},
\end{eqnarray}
where ``const.''~captures other terms not dependent on $\mu_{\rho}({\cal X})$ or $\tilde{\rho}(\cdot)$. This formulation allows us to model the integrated intensity $\mu_{\rho}({\cal X}) $, and the density $ f_{\rho}(\cdot)$ separately. The same approach was also adopted by \citet{Taddy2010} where they developed a nonparametric Dirichlet process mixtures framework for intensity function estimation. The integral $ \int_{{\cal X}} \tilde{\lambda}(x) \log \tilde{\rho}(x) \intd x $ and $\mu_{\lambda}({\cal X})$ are not analytically available since the true intensity function $\lambda(\cdot)$ is unknown. However, treating this integral as an expectation, we see that, for reasonably large $n$, it can be approximated by
\begin{equation} \label{eq:MCapprox}\int_{{\cal X}} \tilde{\lambda}(x) \log \tilde{\rho}(x) \intd x \approx \frac{1}{n} \sum_{i=1}^{n} \log \tilde{\rho}(x_{i}), \end{equation}
where recall that $X \equiv \{x_1,\dots,x_n\}$ is the (observed) point-process realization under the true intensity function $\lambda(\cdot)$. Similarly, by Poissonicity of the NHPP, $n$ is sufficient for $\mu_{\lambda}({\cal X})$, and therefore we approximate the integrated intensity as
\begin{eqnarray}
\label{eq:muLambda}
\mu_{\lambda}({\cal X}) \approx n .
\end{eqnarray}

Using the process-density representation of the intensity function, and the Monte Carlo approximations \eqref{eq:MCapprox} and \eqref{eq:muLambda}, we re-express the optimization problem \eqref{min_kl} in terms of the estimate of the integrated intensity, $\hat{\mu}_\lambda({\cal X})$, and the estimated process density $\hat{\tilde\lambda}(\cdot)$,
\begin{eqnarray}
\label{emp_kl_min}
\{\hat{\mu}_\lambda({\cal X}), \hat{\tilde\lambda}(\cdot) \}= \argmin_{\substack{\mu_{\rho}({\cal X}) \in \mathbb{R}^+ \\ \tilde{\rho}(\cdot) \in {\cal \tilde{A}}}} \left\{  \mu_{\rho}({\cal X}) -  \sum_{i=1}^{n} \log \tilde{\rho}(x_{i}) - n \log \mu_{\rho}({\cal X})  \right\},
\end{eqnarray}
where now $\tilde{\cal A}$ is some set of process densities, which we will establish in Section~\ref{proc_est}. It is easy to see that setting $\mu_{\rho}({\cal X}) = n$ minimizes the objective function in \eqref{emp_kl_min}. Fixing $\mu_{\rho}({\cal X}) = n$ leads us to the optimization problem
\begin{equation}
\label{lambdatildehat}
 \hat{\tilde\lambda}(\cdot) = \argmin_{ \tilde{\rho}(\cdot) \in {\cal \tilde{A}}} \left\{- \sum_{i=1}^{n} \log \tilde{\rho}(x_{i})\right\},
\end{equation}
which is equivalent to maximizing the likelihood of observing $X$. 

\subsection{Modeling the Process Density via Probability Measure}
\label{proc_est}
We model the process density $\tilde{\rho}(\cdot)$ using the transportation of probability measure approach  described in Section \ref{Bg_sec}. Specifically, we seek a diffeomorphism $T: {\cal X} \rightarrow {\cal Z}$, where ${\cal Z}$ need not be the same as ${\cal X}$, such that 
$$ \tilde{\rho}(x) = \eta(T(x)) | \mbox{det}\nabla T(x)|, \quad x \in {\cal X},$$
where $\eta(\cdot)$ is some simple reference density on ${\cal Z}$, and $|\mbox{det}\nabla T(\cdot)| > 0$. Popular choices of $\eta(\cdot)$ include the standard normal distribution if ${\cal Z}$ is unbounded, and the uniform distribution if ${\cal Z}$ is bounded. 
\\\\
While the domain and the range  of the map $T(\cdot)$ can be arbitrary subsets of $\mathbb{R}^{d}$, it is generally easier to construct transport maps from $\mathbb{R}^{d} $ to $\mathbb{R}^{d}$. The domain $\cal X$ is bounded, and therefore we can assume, without loss of generality, that ${\cal X} = (0,1)^{d}$, and we first apply an element-wise \emph{logit} transform to each coordinate of the vector $x = (x^{(1)}, \ldots, x^{(d)})^{'}$ to obtain the vector $y \equiv (y^{(1)},\dots, y^{(d)})' \in \mathbb{R}^d$, where %
$y^{(k)} = \sigma^{-1}(x^{(k)}), \quad k=1, \ldots, d$. The Jacobian of this transformation is given by  $\prod_{k=1}^{d}((x^{(k)})^{-1} + (1-x^{(k)})^{-1})$. 
We subsequently construct the transport map $T(\cdot)$ as a composition of $N$ increasing triangular maps $T_N \circ T_{N-1} \circ \,\cdots\, \circ T_1(\cdot) $ (see Section \ref{comp_sec}). Each triangular map $T_j(\cdot), j=1, \ldots, N,$ in the composition is parameterized using a conditional network approach as detailed in Section \ref{tri_sec}. Specifically, we adopt the neural autoregressive flow where the $k$th component of each triangular map is modeled as in (\ref{uni_flow}).
\\\\


Denote the parameters appearing in the $k$th conditional network associated with the $j$th layer as $\vartheta_{jk}$ and let $\boldsymbol\Theta \equiv \{\theta_{j1}: j = 1,\dots, N \} \cup \{\vartheta_{jk}: k = 2,\dots, d; j = 1,\dots, N\}$. The optimization problem \eqref{lambdatildehat} reduces to the optimization problem:
\begin{eqnarray}
\label{optimTheta}
 \hat{\boldsymbol{\Theta}} = \argmin_{\boldsymbol{\Theta}} \left\{- \sum_{i=1}^{n} \Big( \log \eta( T(y_{i}) ) + \log \mbox{det} \nabla T(y_{i}) \Big)\right\}.
\end{eqnarray}
The optimization problem can be solved efficiently using automatic differentiation libraries, stochastic gradient descent, and graphics processing units for efficient computation. We used \texttt{PyTorch} for our implementation \citep{Paszke2017} and adapted the code provided by \citet{Huang2018}.

\subsection{Universal Approximation}
\label{universal}
The increasing triangular maps constructed using neural autoregressive flows (\ref{uni_flow}) satisfy a universal property in the context of probability density approximation. 
This universal approximation property naturally applies to the process density of a Poisson process. One need only prove this property for the case of one triangular map since, if two maps have the universal property, their composition also has the universal property.

\begin{theorem}
\label{uni_thm}
\Copy{Theorem1}{Let $\cal P$ be a non-homogeneous Poisson process with positive continuous process density $\tilde{\lambda}(\cdot)$ on ${\cal X} \subset \mathbb{R}^{d}$. Suppose further that the weak (Sobolev) partial derivatives of $\tilde{\lambda}$ up to order $d+1$ are integrable over $\mathbb{R}^{d}$.  There exists a sequence of triangular mappings $(T_i(\cdot))_{i}$ wherein the $k$th component of each map $T^{(k)}_i(\cdot)$ has the form $(\ref{uni_flow})$, and wherein the corresponding conditional networks are universal approximators (e.g., feedforward neural networks with sigmoid activation functions), such that
$$ \eta(T_i(\cdot)) \mbox{det} (\nabla T_i(\cdot)) \rightarrow \tilde{\lambda}(\cdot), $$
with respect to the sup norm on any compact subset of $\mathbb{R}^{d}$.  }
\end{theorem}

\noindent We provide a sketch of the proof of Theorem \ref{uni_thm} here and defer the complete proof to Appendix~\ref{proof1sec}. 
\\\\
\emph{Sketch of proof of Theorem \ref{uni_thm}} 

\begin{enumerate}
\item We first establish in Lemma \ref{lemma_exist} in Appendix~\ref{proof1sec} that any process density $\tilde{\lambda}(\cdot)$ that satisfies the conditions in Theorem \ref{uni_thm} can be expressed as $ \tilde{\lambda}(\cdot) = \eta(\tilde{T}(\cdot)) \mbox{det} (\nabla \tilde{T}(\cdot))$, where $\tilde{T}(\cdot)$ is some increasing continuous differentiable triangular map. \\

\item Lemmas \ref{lemma_uni} and \ref{lemma_inv} in Appendix~\ref{proof1sec} establish that, for any increasing continuously differentiable triangular map $\tilde{T}(\cdot)$ and for any $\epsilon >0$, 
$$ |\tilde{T}^{(k)}(x^{(1)}, \ldots, x^{(k)}) - S^{(k)}(x^{(k)}; \tilde{\theta}_k( x^{(1)}, \ldots, x^{(k-1)}))| < \epsilon / 2 ,$$
$$ |\nabla_k \tilde{T}^{(k)}(x^{(1)}, \ldots, x^{(k)}) - \nabla S^{(k)}(x^{(k)}; \tilde{\theta}_k( x^{(1)}, \ldots, x^{(k-1)}))| < \epsilon / 2 ,$$
for $k=2, \ldots, d$, where $\nabla_k$ denotes the derivative with respect to the $k$-th component of the input, and $\tilde{\theta}_k(\cdot)$ is some arbitrary continuous mapping from $(x^{(1)}, \ldots, x^{(k-1)})$ to the set of parameters $a_{ki}, b_{ki}, w_{ki}, i = 1,\dots,M$. \\

\item By the universality of feedforward neural networks with sigmoid activation functions, and uniform continuity of $S^{(k)}(\cdot)$ and $\nabla S^{(k)}(\cdot)$, we have that
$$ |S^{(k)}(x^{(k)}; \tilde{\theta}_k(x^{(1)}, \ldots, x^{(k-1)})) - S^{(k)}(x^{(k)}; \hat{\theta}_k(x^{(1)}, \ldots, x^{(k-1)}))| < \epsilon / 2 ,$$
$$ |\nabla S^{(k)}(x^{(k)}; \tilde{\theta}_k(x^{(1)}, \ldots, x^{(k-1)})) - \nabla S^{(k)}(x^{(k)}; \hat{\theta}_k(x^{(1)}, \ldots, x^{(k-1)}))| < \epsilon / 2 ,$$
for $k=2, \ldots, d$, where $\hat{\theta}_k(\cdot)$ is a feedforward neural network with sigmoid activation functions. \\

\item An application of the triangle inequality yields
  $$  |\tilde{T}^{(k)}(x^{(1)}, \ldots, x^{(k)}) - S^{(k)}(x^{(k)}; \hat{\theta}^{(k)}(x^{(1)}, \ldots, x^{(k-1)}))| < \epsilon,$$
  $$| \nabla_k \tilde{T}^{(k)}(x^{(1)}, \ldots, x^{(k)}) - \nabla S^{(k)}(x^{(k)}; \hat{\theta}_k(x^{(1)}, \ldots, x^{(k-1)}))| < \epsilon ,$$
for $k=2, \ldots, d$. Therefore, for an arbitrary increasing continuous differentiable triangular map $\tilde{T}(\cdot)$, there exists a triangular map $T(\cdot)$ where the $k$-th component of $T(\cdot)$ is $S^{(k)}(x^{(k)};\hat{\theta}_k(x^{(1)}, \ldots, x^{(k-1)}))$ such that
$$ |T^{(k)}(x^{(1)}, \ldots, x^{(k)}) - \tilde{T}^{(k)}(x^{(1)}, \ldots, x^{(k)})| < \epsilon, $$
$$ |\nabla_k T^{(k)}(x^{(1)}, \ldots, x^{(k)}) - \nabla_k \tilde{T}^{(k)}(x^{(1)}, \ldots, x^{(k)})| < \epsilon,$$
for $k=2, \ldots, d$. We naturally also have that
$$ |T^{(1)}(x^{(1)}) - \tilde{T}^{(1)}(x^{(1)})| < \epsilon, $$
$$ |\nabla_1 T^{(1)}(x^{(1)}) - \nabla_1 \tilde{T}^{(1)}(x^{(1)})| < \epsilon,$$
where $T^{(1)}(x^{(1)}) = S^{(1)}(x^{(1)}; \hat{\theta}_1)$.\\

\item Finally, since $\mbox{det} (\nabla T(x)) = \prod_{k=1}^{d} \nabla_k T^{(k)}(x)$, and the target density $\eta(\cdot)$ is smooth, both $| \mbox{det} (\nabla T(x)) - \mbox{det} (\nabla \tilde{T}(x)) |$, and $| \eta( T(x) ) - \eta( \tilde{T}(x) ) |, $ can be made arbitrarily small for all $x \in \mathcal{X}$. This implies that
  $$ |\eta(T(x)) \mbox{det}( \nabla T(x)) - \eta( \tilde{T}(x)) \mbox{det}  (\nabla \tilde{T}(x)) |, $$
can be made arbitrarily small  for all $x \in \mathcal{X}$, which completes the proof.
\end{enumerate}

\subsection{Simulating from the fitted point process}

An attraction of using measure transport is that one can readily simulate data based on the estimated intensity function without resorting to methods like thinning, which can be inefficient when the Poisson process is highly non-homogeneous. Here, one simulates from the simple, known reference density $\eta(\cdot)$, and then \emph{pushes back} the points through the inverse map. Since the maps we use are triangular, their inverse can be found in a relatively straightforward manner. 

Consider a point $z=(z^{(1)}, \ldots, z^{(d)})^{'} $ in the reference domain. We give an algorithm for computing $T^{-1}(z)$, where $T(\cdot)$ is a (single) increasing triangular map,  in Algorithm \ref{algo_inv}. Note that inversion under the increasing triangular map involves solving $d$ univariate root-finding problems. These problems can be efficiently solved  since each component of the map is continuous and increasing.

\begin{algorithm}[t!]
\caption{Triangular Map Inversion \newline \textbf{Input: } $z \in \mathbb{R}^{d}$, triangular map $T(\cdot)$  \newline \textbf{Output: } $y \in \mathbb{R}^{d}$ }
\begin{algorithmic}
  \State \texttt{Find $y^{(1)}$ such that $T^{(1)}(y^{(1)}) = z^{(1)}$} 
 \For{\texttt{$k=2,\ldots,d$}}
  \State \texttt{Find $y^{(k)}$ such that $T^{(k)}(y^{(1)},\ldots, y^{(k-1)},y^{(k)}) = z^{(k)}$}
\EndFor
 \State \texttt{Return: $y = (y^{(1)}, \ldots, y^{(d)})^{'}$}
\end{algorithmic}
\label{algo_inv}
\end{algorithm}

Now, when we have $N$ triangular maps in composition, $T_N \circ T_{N-1} \circ \cdots \circ T_1(\cdot)$ say, we can compute the inverse by iteratively applying Algorithm \ref{algo_inv} using $T_{N}(\cdot), T_{N-1}(\cdot), \ldots, T_1(\cdot)$. Algorithm \ref{algo_cond_sim} gives an algorithm for simulating from the fitted point process for the case when the reference probability measure is the standard multivariate normal distribution.

\begin{algorithm}[t!]
\caption{Point Process Simulation \newline \textbf{Input: } number of points $n$, maps $T_N(\cdot), \dots, T_1(\cdot)$  \newline \textbf{Output: } simulated points $x_1, \ldots, x_n \in \mathbb{R}^{d}$ }
\begin{algorithmic}
\For{\texttt{$i=1,\ldots, n$}}
  \State \texttt{Draw $z_i \sim {\cal N}(0,I_d)$}
  \State \texttt{Apply Algorithm \ref{algo_inv} to compute the inverse $y_i$ of $z_i$ under $T_N \circ \cdots \circ T_1 (\cdot)$}
  \State \texttt{Set $x_{i}^{(k)} = \sigma(y^{(k)}_{i}), \quad k=1,\ldots,d$ }
\EndFor
\State \texttt{Return: $\{x_1, \ldots, x_{n}\}$}
\end{algorithmic}
\label{algo_cond_sim}
\end{algorithm}
 
\subsection{Standard Error Estimation}

The intensity function is fitted by solving the problem given in \eqref{optimTheta}. The standard error of the fitted intensity function can be estimated using a non-parametric bootstrapping approach \citep{Efron1981}. We construct $B$ bootstrap samples by drawing the number of points $n_b, b = 1,\dots,B$, from a Poisson distribution with rate parameter $n$. For each $b = 1,\dots,B$ we then randomly sample $n_b$ points from the observed points $X$ with replacement, and fit the process density to each of these bootstrap samples, to obtain $B$ estimated process densities $\hat{\tilde{\lambda_b}}(\cdot), b = 1,\dots, B$. The standard error of the intensity function evaluated at any point $x \in {\cal X}$ is then obtained by finding the empirical standard deviation of $\{n_b\hat{\tilde\lambda}_b(x): b = 1,\dots,B\}$. Algorithm \ref{algo_boot_np} gives a summary.
\\\\
Standard error estimation can also be performed using a parametric bootstrap approach \citep{Efron1979}, where bootstrap samples are obtained from the fitted intensity function. However, this would require running Algorithms \ref{algo_inv} and \ref{algo_cond_sim} $B$ times, which would be considerably more computationally demanding. We therefore do not consider this bootstrap strategy here.

\begin{algorithm}[t!]
\caption{Standard Error Estimation Using Non-Parametric Bootstrapping \newline \textbf{Input: } observational data $X$, 
bootstrap sample size $B$ \newline \textbf{Output: } bootstrap estimates of the intensity function}
\begin{algorithmic}

\For{\texttt{$b=1, \ldots, B$}}
   \State \texttt{Draw $n_b \sim \mbox{Poisson}(n)$}
   \State \texttt{Sample $n_b$ points with replacement from $X$ to obtain $X_b \equiv \{x_1, \ldots, x_{n_b}\}$}
   \State \texttt{Estimate the triangular map, and consequently the process density, $\hat{\tilde\lambda}_b(\cdot)$ using the\\ \qquad\qquad bootstrap sample $X_b$}.
\EndFor
\State \texttt{Return: Bootstrap samples of the intensity function, $\{n_b \hat{\tilde{\lambda}}_b(\cdot)\}_{b=1}^{B}$ }
\end{algorithmic}
\label{algo_boot_np}
\end{algorithm}

\section{Illustrations}
\label{experiments}
In this section we illustrate the application of our proposed method through simulation experiments (Section~\ref{sim_sec}) and in the context of earthquake intensity modeling (Section~\ref{sec_quake}). The purpose of the simulation experiments is to demonstrate the validity of our approach and to explore the sensitivity of the estimates to the conditional networks' structure. All our illustrations are in a one or two dimensional setting, as these cover the majority of applications, but our method is scalable to higher dimensions due to the map's triangular structure should this be needed.

\subsection{Simulation Experiments}
\label{sim_sec}
 In this section we illustrate our method on simulated data in both a one dimensional and a two dimensional setting.  Our method requires one to specify the number of compositions of triangular maps, the width of the neural network in the triangular maps, the number of layers in the conditional networks (feedforward neural networks), and the width of each layer in each conditional network. In neural networks and deep learning literature, choosing the optimal network structure is an open problem. For shallow feedforward neural networks, that is, neural networks with one or two hidden layers, information criteria based methods \citep{Fogel1991} and heuristic algorithms \citep{Leung2003, Tsai2006} have been proposed to determine the optimal width of the network. However, to the best of our knowledge, analogous methods are not available for deep neural networks and neural networks with complicated structure. In higher dimensional problems, there is theoretical support for adopting a very deep neural network structure due to its representational power \citep{Eldan2016, Raghu2017}. However, since point process realizations typically lie in lower dimensional space, such results are less relevant.
\\\\
We found that in low-dimensional settings our estimates did not change considerably with the number of layers in the conditional network and therefore, here, we fix the number of layers to one. That is, we let each $\theta_{jk}^{(i)}(\cdot) $ be the output of a neural network of one layer. In both simulation experiments we set the widths of the neural networks in both the triangular maps (i.e., $M$) and the conditional networks, to 64. We set this number by successively increasing the network widths in powers of two until the intensity-function estimate was not improved. We also used network widths of 64 in the application case study of Section~\ref{sec_quake}. The learning rate for the optimization is set to $10^{-4}$ in all our experiments.
\\
\\
For the one dimensional studies, we first simulated events (via thinning) from the following two one-dimensional intensity functions,
\begin{align}
\label{sim1_intensity1}
 \lambda_1(x^{(1)}) &= 500 + 300 \sin(10x^{(1)}), \quad 0 < x^{(1)} < 1, \\
\label{sim1_intensity2}
 \lambda_2(x^{(1)}) &= 500, \quad 0 < x^{(1)} < 1. 
\end{align}
We then fitted several models to the simulated events, each with a different number of compositions of triangular maps. The procedure of simulating and model fitting was repeated 40 times in order to assess the variability in the estimated intensity functions. Each model fitting required approximately two minutes on a graphics processing unit.
\\
\\
The average and empirical standard deviation of the $L_2$ distance between the true intensity function and the estimated intensity function are shown in Tables \ref{sim1_table} and \ref{sim1_2_table}. Reassuringly, we see that the estimates, and the variability thereof, are consistent across different numbers of compositions of triangular maps for both cases. For reference, we also provide the results from kernel density estimation, where the bandwidth of a Gaussian kernel was chosen using Silverman's rule of thumb \citep{Silverman1986}. In this experiment we observe significant improvement in using deep compositional maps over conventional kernel density estimation.  For illustration, the estimated intensity functions using four compositions of triangular maps for both intensity functions are shown in Figure \ref{sim1_plot}. 
\\
\begin{table}[t!]
\caption{Average and empirical standard deviation of the $L_2$ distance between the true and the fitted intensity functions in the one dimensional simulation experiment with intensity function $ \lambda(x^{(1)}) = 500 + 300 \sin(10x^{(1)}), \quad 0 < x^{(1)} < 1$.}
\begin{center}
\begin{tabular}{lcccccc}
   \hline
No. of compositions of triangular maps & 1 & 2& 3 & 4 &5 & KDE  \\ \hline\hline
Average $L_2$ distance & 77.5 & 77.4 & 76.5 & 77.6 & 78.3 & 101.2\\ 
Standard deviation of $L_2$ distance & 12.5 & 11.3 &10.3 & 12.2 & 13.1 & 15.6 \\ \hline\hline 
\end{tabular}
\end{center}
\label{sim1_table}
\end{table}

\begin{table}[t!]
\caption{Average and empirical standard deviation of the $L_2$ distance between the true and the fitted intensity functions in the one dimensional simulation experiment with intensity function $ \lambda(x^{(1)}) = 500, \quad 0 < x^{(1)} < 1.$ }
\begin{center}
\begin{tabular}{lcccccc}
   \hline
No. of compositions of triangular maps & 1 & 2& 3 & 4 &5 & KDE  \\ \hline\hline
Average $L_2$ distance & 51.8 & 50.9 & 50.6 & 50.5 & 50.8 & 81.3\\ 
Standard deviation of $L_2$ distance & 12.9 & 13.0 &13.0 & 13.3 & 13.1 & 9.9 \\ \hline\hline 
\end{tabular}
\end{center}
\label{sim1_2_table}
\end{table}

\begin{figure*}
  \centering
  \includegraphics[width = 0.45\linewidth]{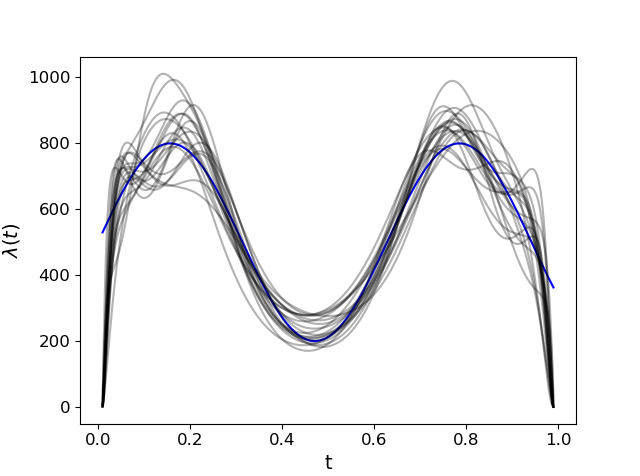} 
 \includegraphics[width = 0.45\linewidth]{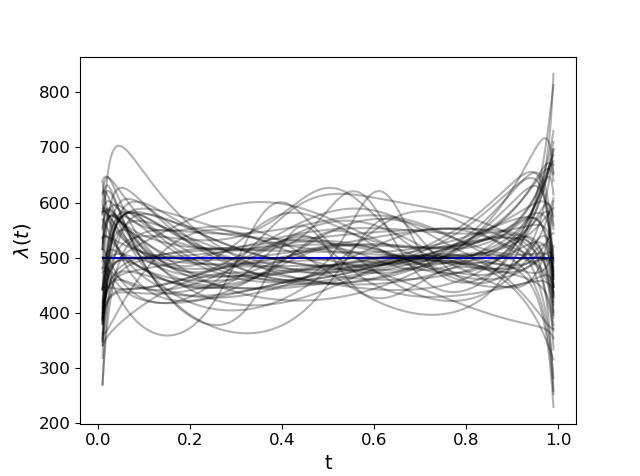} 
  \caption{The true intensity function \eqref{sim1_intensity1} (left) and \eqref{sim1_intensity2} (right) (blue) from which 40 point patterns were independently simulated from, and the corresponding 40 estimated intensity functions using measure transport with four compositions of triangular maps (black).}
  \label{sim1_plot}
\end{figure*}

For the two-dimensional simulation studies, we simulated data from the following two intensity functions:
\begin{align}
  \lambda_3(x^{(1)}, x^{(2)}) &= (30 + 10 \sin(10x^{(1)}))(30 + 10 \cos(20x^{(2)})), \nonumber \\
  &\qquad\qquad\qquad\qquad\qquad\qquad 0 < x^{(1)} < 1,\quad 0 < x^{(2)} < 1,\label{sim2_intensity1}\\
\label{sim2_intensity2}
 \lambda_4(x^{(1)}, x^{(2)}) &= 900, \quad 0 < x^{(1)} < 1, \quad 0 < x^{(2)} < 1.
\end{align}
We used the same procedure as in the one dimensional case study whereby we fitted each model to the events simulated from the intensity function. We again simulated and fit 40 times to assess the variability of our estimates. For this experiment we slightly enlarged the domain associated with the point process to reduce problems related to boundary effects. Each model-fitting required approximately ten minutes on a graphics processing unit. 
\\
\\
The average and empirical standard deviation of the $L_2$ distance between the true intensity function and the estimated intensity function are shown in Table \ref{sim2_table} and \ref{sim2_2_table}. As in the one dimensional case, we do not observe substantial differences in the estimates when the number of compositions is varied, and also that the measure transport approach substantially outperformed  kernel density estimation. The true intensity surface, together with the average and standard deviation of the estimated intensity surfaces across the 40 simulations for the intensity functions \eqref{sim2_intensity1} and \eqref{sim2_intensity2}, for the case of four compositions of triangular maps, are shown in Figures \ref{sim2_plot} and \ref{sim2_2_plot}, respectively.  We see from the plots that the proposed method manages to recover the true intensity surface on average, and that the variability in the estimation is large when the true intensity is large. This is expected since the variance of a Poisson random variable is proportional to its mean.
\\\\
In summary, these experiments illustrate that our method based on measure transport is computationally efficient, and that it does not overfit as the number of compositions of triangular maps increases. We have also observed substantial improvement over kernel density estimation for intensity function estimation. Choosing the number of compositions using a formal model selection approach is desirable; however, as quantifying the complexity of the proposed model is difficult, various information criteria such as the Bayesian Information Criterion (BIC) are not applicable. We also compare the running times when fitting the models using a mid-end graphics processing unit (GPU) and a high-end CPU; see Table \ref{comp_time_table}. It is clear that there is a substantial computational benefit to using a GPU when training these models.

\begin{table}[t!]
\caption{Average and empirical standard deviation of the $L_2$ distance between the true and the fitted intensity functions in the two dimensional simulation experiment with intensity function $ \lambda(x^{(1)}, x^{(2)}) = (30 + 10 \sin(10x^{(1)}))(30 + 10 \cos(20x^{(2)})), \quad 0 < x^{(1)} < 1,\quad 0 < x^{(2)} < 1 .$ }
\begin{center}
\begin{tabular}{lcccccc}
   \hline
No. of compositions of triangular maps & 1 & 2& 3 & 4 &5 & KDE  \\ \hline\hline
Average $L_2$ distance & 271.4 & 267.1 & 272.2 & 273.4 & 278.3 & 438.7 \\ 
Standard deviation of $L_2$ distance & 36.7 & 37.9 & 36.3 & 34.6 & 36.7 & 14.1 \\ \hline\hline 
\end{tabular}
\end{center}
\label{sim2_table}
\end{table}

\begin{figure*}[!t]

\centering

  \includegraphics[width=2.4in]{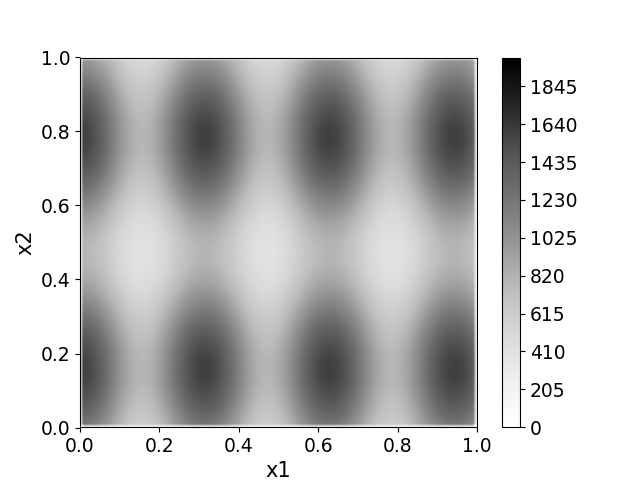}\hspace{.1em}%
  \includegraphics[width=2.4in]{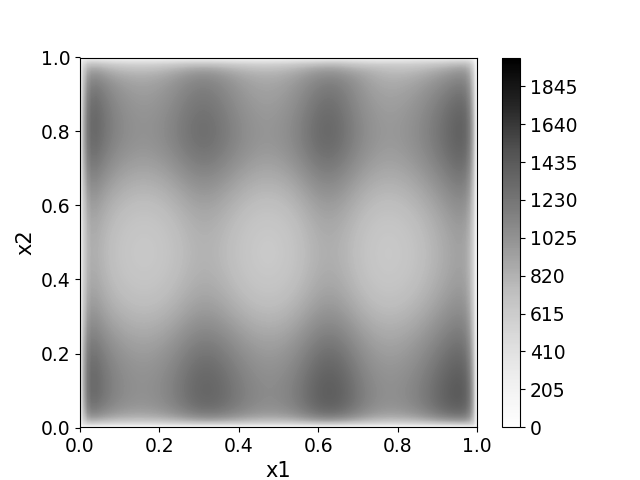}\hspace{.1em}%
  \includegraphics[width=2.4in]{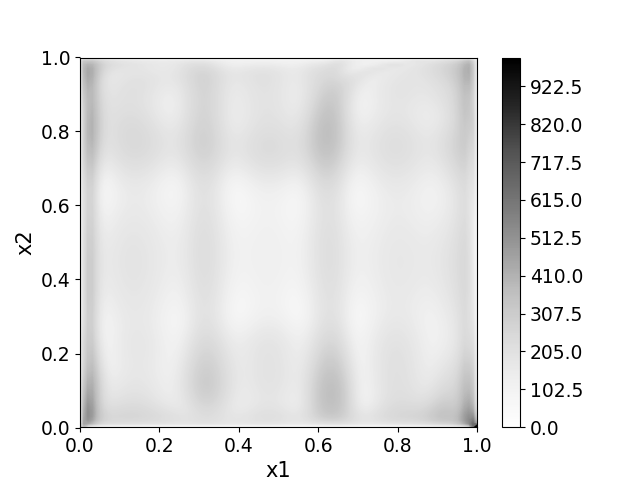}\hspace{.1em}%

\caption{Top-left panel: The true intensity surface \eqref{sim2_intensity1} used to generate 40 point patterns in the two dimensional experiment. Top-right panel: Average estimated intensity surface. Bottom panel: Empirical standard deviation of the estimated intensity surfaces. }
\label{sim2_plot}
\end{figure*}

\begin{table}[t!]
\caption{Average and empirical standard deviation of the $L_2$ distance between the true and the fitted intensity functions in the two dimensional simulation experiment with intensity function $\lambda(x^{(1)}, x^{(2)}) = 900, \quad 0 < x^{(1)} < 1, \quad 0 < x^{(2)} < 1$.}
\begin{center}
\begin{tabular}{lcccccc}
   \hline
No. of compositions of triangular maps & 1 & 2& 3 & 4 &5 & KDE  \\ \hline\hline
Average $L_2$ distance & 145.7 & 144.6 & 141.2 & 145.7 & 144.9 & 227.9 \\ 
Standard deviation of $L_2$ distance & 7.5 & 7.2 & 7.1 & 8.4 & 7.3 & 14.2 \\ \hline\hline 
\end{tabular}
\end{center}
\label{sim2_2_table}
\end{table}

\begin{figure*}[!t]

\centering

  \includegraphics[width=2.4in]{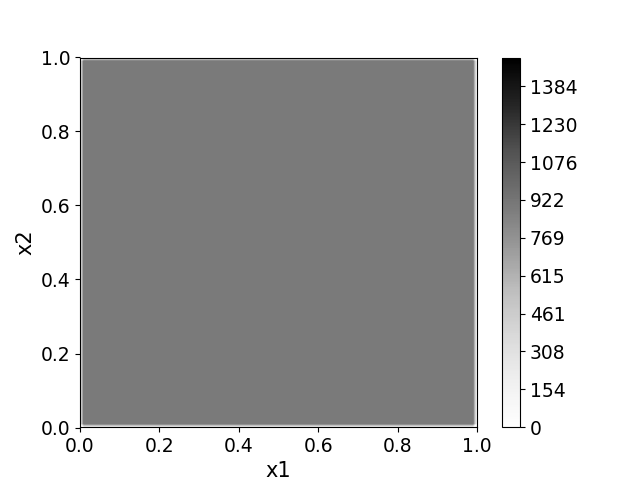}\hspace{.1em}%
  \includegraphics[width=2.4in]{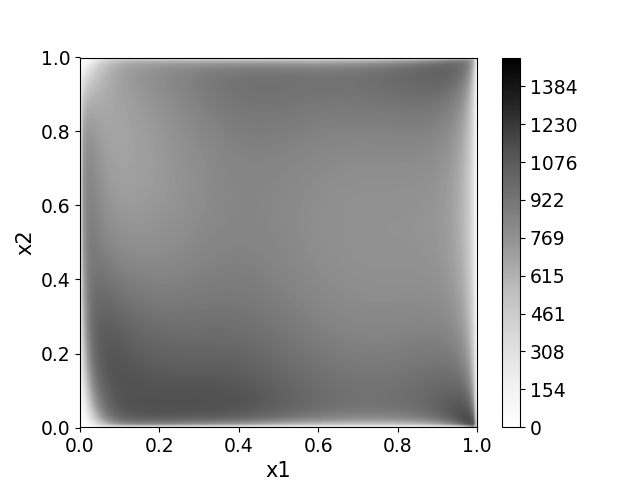}\hspace{.1em}%
  \includegraphics[width=2.4in]{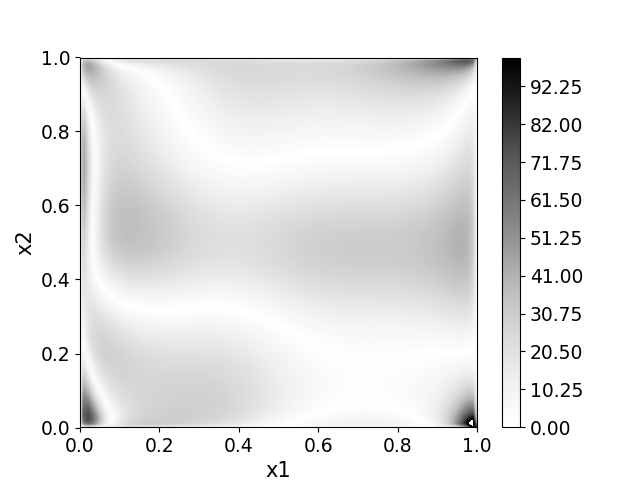}\hspace{.1em}%

\caption{Top-left panel: The true intensity surface \eqref{sim2_intensity2} used to generate 40 point patterns in the two dimensional experiment. Top-right panel: Average estimated intensity surface. Bottom panel: Empirical standard deviation of the estimated intensity surfaces. }
\label{sim2_2_plot}
\end{figure*}

\begin{table}[t!]
\caption{Average computational time when fitting the model using a GPU (NVIDIA GeForce GTX 1080Ti) and a CPU (Intel(R) Core(TM) i9-7900X CPU @ 3.30GHz).}
\begin{center}
\begin{tabular}{lcccc}
   \hline
Intensity function & $\lambda_1$ & $\lambda_2$ & $\lambda_3$ & $\lambda_4$  \\ \hline\hline
GPU Time (seconds) & 39.95 & 39.42 & 57.74 & 58.23  \\  
CPU Time (seconds) & 80.92 & 57.33 & 146.72 & 137.48  \\  \hline\hline 
\end{tabular}
\end{center}
\label{comp_time_table}
\end{table}

\subsection{Modeling Earthquake Data}
\label{sec_quake}
In this section we apply our method for intensity function estimation to an earthquake data set comprising 1000 seismic events of body-wave magnitude (MB) over 4.0. The data set is available from the \texttt{R} \emph{datasets} package. The events we analyze are those that occurred near Fiji from 1964 onwards. The left panel of Figure \ref{quakes_est} shows a scatter plot of locations of the observed seismic events. 
\\\\
We fitted our model using a composition of five triangular maps. The estimated intensity surface and the standard error surface obtained using Algorithm~\ref{algo_boot_np} are shown in the middle and right panels of Figure \ref{quakes_est}, respectively. As was observed in the simulation experiments, we see that the estimated standard error is large in areas where the estimated intensity is high. The probability that the intensity function exceeds various threshold can also be estimated using non-parametric bootstrap resampling; some examples of these exceedance probability plots are shown in Figure \ref{quakes_exceed_plot}.
\\\\

\begin{figure*}[!t]
\centering
\includegraphics[width=2.45in]{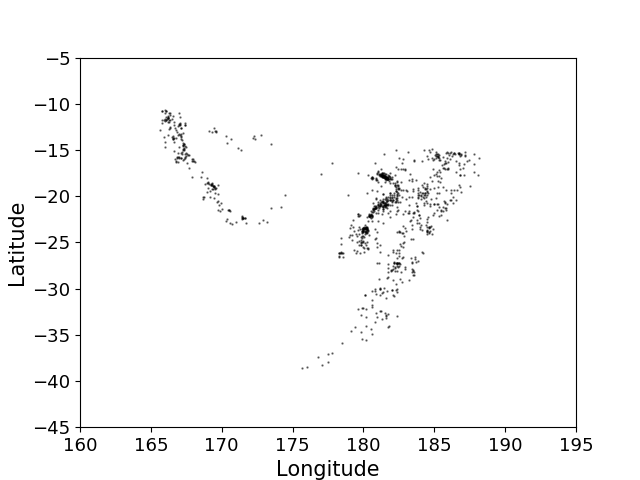}\hspace{.1em}%
\includegraphics[width=2.45in]{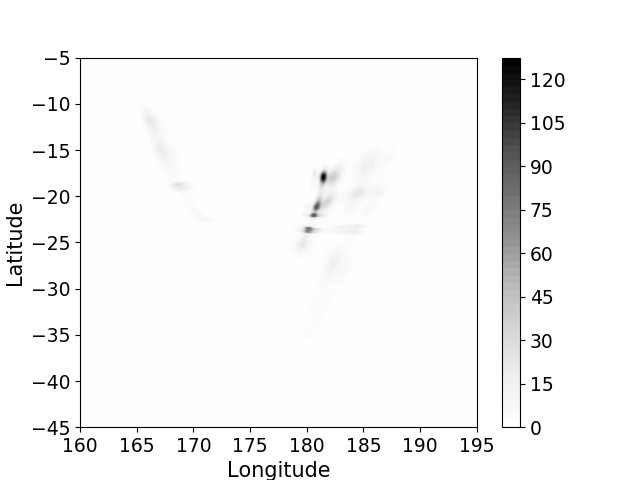}\hspace{.1em}%
\includegraphics[width=2.45in]{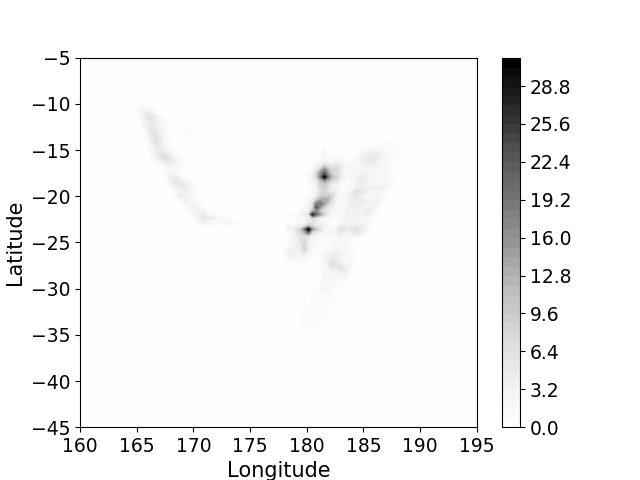}\hspace{.1em}%
\caption{Top-left panel: Scatter plot of earthquake events with body-wave magnitude greater than 4.0 near Fiji since 1964. Top-right panel: Estimated intensity function obtained using measure transport.  Bottom panel: Estimated standard error of the intensity surface obtained using Algorithm~\ref{algo_boot_np}.}
\label{quakes_est}
\end{figure*}

\begin{figure*}[!t]
\centering

  \includegraphics[width=2.45in]{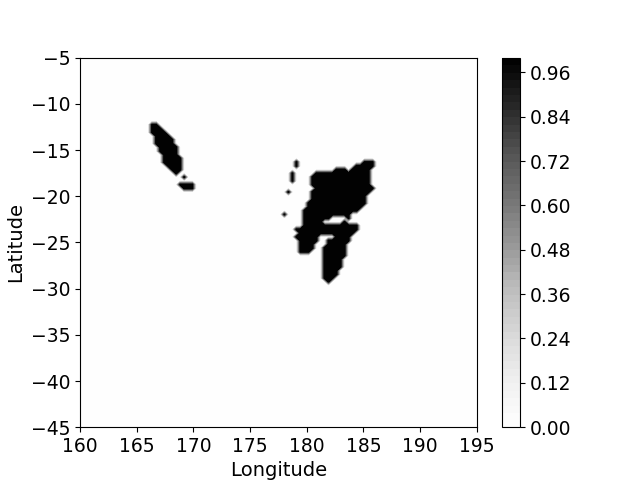}\hspace{.1em}%
  \includegraphics[width=2.45in]{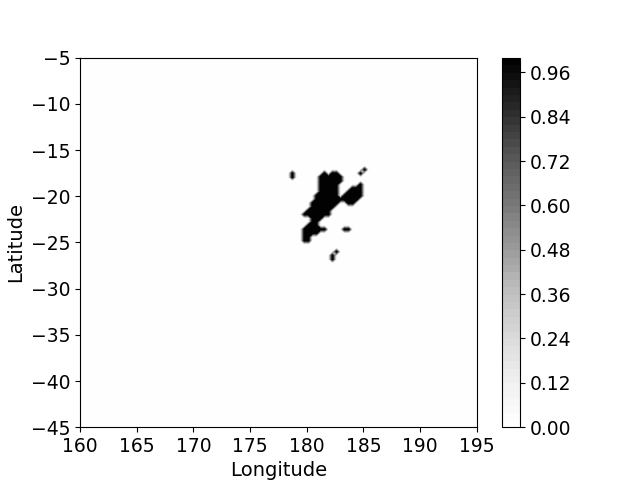}\hspace{.1em}%
  \includegraphics[width=2.45in]{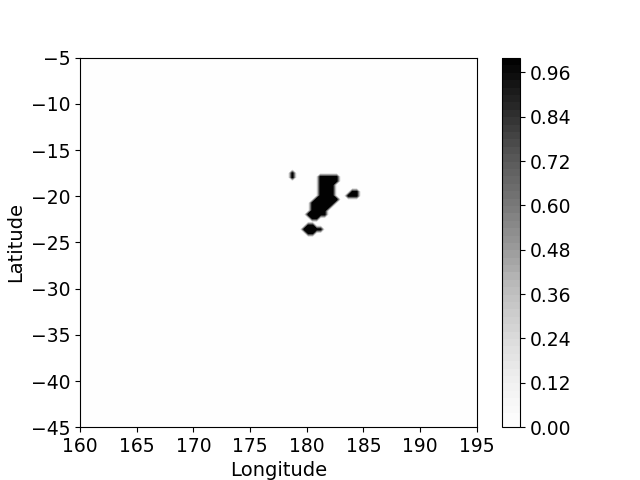}\hspace{.1em}%

\caption{Top-left panel: Estimated exceedance probability $P(\lambda(\cdot) > 1)$. Top-right panel: Estimated exceedance probability $P(\lambda(\cdot) > 5)$. Bottom panel: Estimated exceedance probability $P(\lambda(\cdot) > 10)$. }
\label{quakes_exceed_plot}
\end{figure*}

A ubiquitous model used in such applications is the log-Gaussian Cox process (LGCP). For comparative purposes, here we fit an LGCP using the package \texttt{inlabru} \citep{Bachl_2019}, with the latent Gaussian process equipped with a constant (unknown) mean and a Mat{\'e}rn covariance function with smoothness parameter $\nu = 1$. The Gaussian process was approximated via a stochastic partial differential equation \citep{Lindgren_2011} on a mesh comprising 2482 vertices. Approximate inference and prediction required about two minutes on a fine grid comprising 15262 pixels.
\\\\
For both our intensity-function estimate, and the posterior median intensity function from the LGCP, we compute a QQ-plot. In this QQ-plot we use the horizontal axis to represent the fitted quantiles from the density estimate and the vertical axis to represent the empirical quantiles obtained from the observational data; the identity line is used to denote perfect agreement. We see from Figure~\ref{KS_plot} that the intensity function estimates using both models are reasonable, with that obtained using measure transport slightly better. The Kolmogorov-Smirnov statistic for our approach is 0.039 while that from the LGCP is 0.048.

\begin{figure*}
  \centering
  \includegraphics[width = 0.5\linewidth]{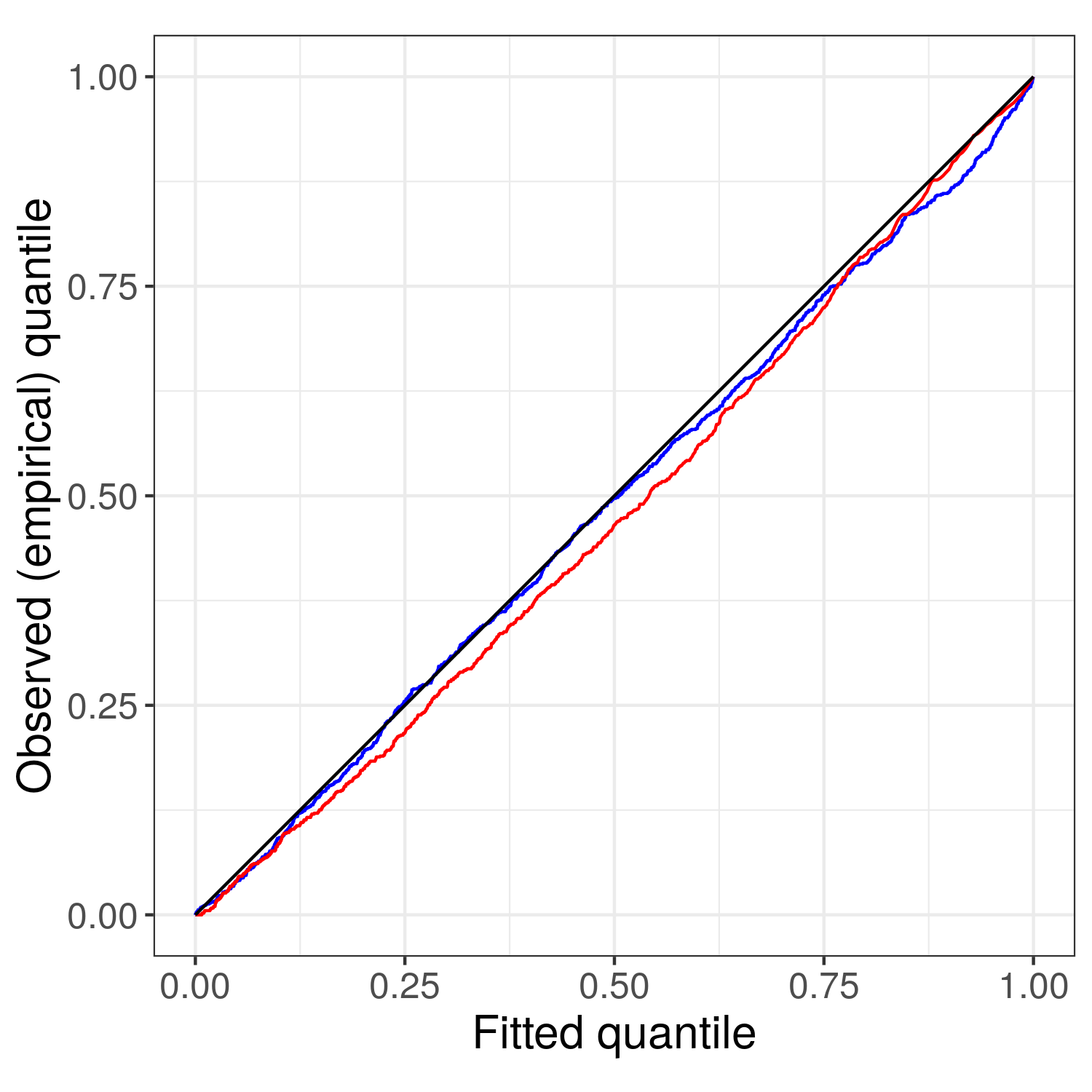}
  \caption{QQ-plot comparing the fitted quantiles (from the intensity function estimates) to the observed empirical quantiles. The blue line corresponds to the intensity function estimate obtained from measure transport, the red line to the posterior median from an LGCP fitted using \texttt{inlabru}. The black line denotes perfect fit. }
  \label{KS_plot}
\end{figure*}

\section{Conclusion}
\label{conc_sec}
This paper develops a general and scalable approach to the problem of modeling and estimating the intensity function of a non-homogeneous Poisson process. We leverage the measure transport framework through compositions of triangular maps to model the unknown intensity function, and utilize software libraries originally created for deep learning for efficient inference. The developed model is shown to have the universal property whereby any positive continuous intensity function can be approximated arbitrarily well. 
\\\\
Our experiments clearly demonstrate the practical advantage of the measure transport approach over simpler methods such as KDE. The performance of the proposed method is also competitive as compared to the use of LGCPs. Furthermore, the measure transport approach has other amenable properties. Notably, the use of a simple reference density allows one to easily simulate point processes, and back transform the coordinates to the original space, with little effort. This leads to an efficient simulation algorithm, as well as an efficient bootstrap algorithm for uncertainty quantification. Second, our approach has the potential to recover spatial properties (such as anisotropy and nonstationarity), that would require additional modeling effort with models such as the LGCP, or more sophisticated kernels with KDE. Finally, our approach is highly scalable, and can be extended to higher dimensional spaces with no modification to the underlying software.
\\\\
There are several possible avenues for future work. First, in this work we have only considered low-dimensional spatial problems. However, the measure transport approach naturally extends to higher-dimensional spaces. For spatio-temporal point processes, for example, one could simply add an additional, temporal, dimension to the two-dimensional spatial model. Second, a simple way to incorporate covariate information, which is not as straightforward as in an LGCP, say, will be important for the approach to find wide applicability in a practical setting. 
\\\\
Code to reproduce the results in the simulation and real-data illustrations is available as supplementary material.

\section*{Acknowledgements}

AZ--M was supported by the Australian Research Council (ARC) Discovery Early Career Research Award, DE180100203. The authors would like to thank Noel Cressie for helpful discussions on bootstrapping. The authors would like to thanks the editors and reviewers for helpful suggestions which have significantly improved the paper.

\appendix
\section{Proof of Results}

\label{Appendix}

\subsection{Proof of Proposition \ref{kl_prop}}
\label{proofKLsec}


\begin{proof}
By definition of the KL divergence,
$$  D_{KL}(f_{\rho_{1}}||f_{\rho_{2}}) = \mathbb{E}_{\rho_{1}} \bigg\{  \log \frac{ f_{\rho_{1}}(x)}{ f_{\rho_{2}}(x) } \bigg\}, $$
where $\mathbb{E}_{\rho_1}\{\cdot\}$ is the expectation taken with respect to the density $f_{\rho_1}(\cdot)$. By Campbell's theorem, we have that 
$$ \mathbb{E}_{\rho_i} \Big\{ \sum_{x \in {\cal P}_i} \log {\rho_j}(x) \Big\} = \int_{{\cal X}} (\log {\rho_j}(x) ) \rho_i(x) \intd x, \quad i,j = 1,2.$$ 
Therefore, from \eqref{fdensity},  
$$  \mathbb{E}_{\rho_{1}} \{ \log f_{\rho_{1}}(x) \} = - \int_{\cal X} (\rho_{1}(x) - 1) \intd x +  \int_{\cal X} (\log \rho_{1}(x) ) \rho_{1}(x) \intd x, $$
$$ \mathbb{E}_{\rho_{1}} \{ \log f_{\rho_{2}}(x) \} = - \int_{\cal X} (\rho_{2}(x) - 1) \intd x +  \int_{\cal X} (\log \rho_{2}(x) ) \rho_{1}(x) \intd x.$$
Combining these two equalities completes the proof.
\end{proof}

\subsection{Proof of Theorem \ref{uni_thm}}
\label{proof1sec}
The first lemma we need is Lemma 2.6 of \cite{Bogachev2005} stated in a slightly different form. 
\begin{lemma}
\label{lemma_exist}
Suppose the probability measures $\mu(\cdot)$ and $\nu(\cdot)$ on $\mathbb{R}^{d}$ are given by continuous positive densities $\rho_{\mu}(\cdot)$ and $\rho_{\nu}(\cdot)$, respectively, whose weak (Sobolev) partial derivatives up to order $d+1$ are integrable over $\mathbb{R}^{d}$. Then there exists an increasing continuously differentiable triangular mapping $\tilde{T}_{\#\mu}(\cdot)$ such that
$$ \tilde{T}_{\#\mu}(\cdot) = \nu(\cdot) .$$
\end{lemma}
Lemma \ref{lemma_exist} implies that it is sufficient to consider the space of increasing continuous differentiable triangular maps when one is seeking to push forward the measure $\mu(\cdot)$ to another, usually simpler, reference measure $\nu(\cdot)$. In this work, we fix the reference measure $\nu(\cdot)$ to be standard multivariate Gaussian distribution. The following two lemmas show that the triangular maps constructed using the neural autoregressive flows are indeed dense in the space of increasing continuous differentiable triangular maps.

\begin{lemma}
\label{lemma_uni}
The set of functions $$\left\{h: \mathbb{R} \rightarrow (0,1), h(x) = \sum_{i=1}^M w_i \sigma( a_i x  + b_i) \bigg|  M \in \mathbb{N};\, a_i > 0~ \forall i;\, b_i \in \mathbb{R}~ \forall i;\, w_i >0~ \forall i;\, \sum_{i=1}^{M} w_i = 1 \right\},$$ is dense in the space of monotonically increasing continuous differentiable functions $f: \mathbb{R} \rightarrow (0,1)$ with $f(t) \rightarrow 0$ as $t \rightarrow -\infty$ and $f(t) \rightarrow 1$ as $t \rightarrow \infty$ with respect to the norm
$$ || f ||_{ {\cal C}^{1}(I) } := \max_{k=0,1} \max_{t \in I} |f^{[k]}(t) | ,$$
on compact intervals $I=[I_0,I_1] \subset \mathbb{R}$.
\end{lemma}
\begin{proof}
Fix a sufficiently small $\epsilon > 0$. Let $f: \mathbb{R} \rightarrow (0,1)$ be a monotonically increasing ${\cal C}^{1}$ function with $f(t) \rightarrow 1$ as $t \rightarrow \infty$ and $f(t) \rightarrow 0$ as $t \rightarrow - \infty$. Therefore, $f'(t)$ is a positive continuous probability density function on $\mathbb{R}$. Now, for any $a_i >0$ and $b_i \in \mathbb{R}$, $\sigma_i(t) := \sigma(a_i t + b_i)$ satisfies $\sigma_i(t) \rightarrow 1$ as $t \rightarrow \infty$ and $\sigma_i(t) \rightarrow 0$ as $t \rightarrow -\infty$. Therefore, $\sigma'_{i}(t)$ is a positive continuous density on $\mathbb{R}$. Therefore, by standard results in approximation theory (see, e.g., Section 4 of \cite{Nestoridis2007}), for any $\epsilon' >0$ and any compact interval $K \subset \mathbb{R}$, there exists 
$$ h(t) = \sum_{i=1}^{N} w_i \sigma(a_i t + b_i) ,$$
for some $a_i > 0, b_i \in \mathbb{R}, i=1, \ldots, N,$ such that
$$ |f'(t) - h'(t)| < \epsilon' ,$$
for any $t \in K$. In particular, the above result follows from Lemma 3.1 of \cite{Nestoridis2007}, and we note that for positive continuous probability density $h$, the weights $w_i$ can be chosen such that $\sum_{i=1}^{N} w_i = 1$.
\\\\
Fix $s_0, s_1$ such that $f(s_0) = \epsilon$ and $f(s_1) = 1 - \epsilon$. For $\epsilon$ sufficiently small we have that $I \subset [s_0, s_1]$. Let $\epsilon' = \epsilon / (s_1 - s_0)$, we have that $|f'(t) - h'(t)| < \epsilon' $ for all $t \in [s_0, s_1]$. Therefore,
$$ \bigg| \int_{s_0}^{s_1} f'(s) - h'(s) \intd s \bigg| \le \int_{s_0}^{s_1} |f'(s) - h'(s)| \intd s < \epsilon .$$
Using the inequality above along with $f(s_1) - f(s_0) = 1 - 2 \epsilon$, it is straightforward to deduce that $f(s_0) < 3 \epsilon$. Now, for any $t \in I$, we have
\begin{eqnarray*}
 |f(t) - h(t) | &=& \bigg| \int_{-\infty}^{t} f'(t) - h'(t) \intd t  \bigg| \\
         & \le & \bigg| \int_{-\infty}^{s_0} f'(t) - h'(t) \intd t  \bigg| + \bigg| \int_{s_0}^{s} f'(t) - h'(t) \intd t \bigg|   \\
         & \le & 2 \epsilon + \epsilon' (s - s_0) \\
         & \le & 3 \epsilon .
\end{eqnarray*}
Therefore, we have that $ \| f(\cdot) - h(\cdot) \|_{ {\cal C}^{1}(I) } \le 3 \epsilon $.
\end{proof}

\begin{lemma}
\label{lemma_inv}
The set of functions
$$\left\{g: \mathbb{R} \rightarrow \mathbb{R},  g(x) := \sigma^{-1} \Big( \sum_{i=1}^{M} w_i \sigma( a_i x + b_i) \Big) \bigg|  M \in \mathbb{N}, a_i > 0~ \forall i;\, b_i \in \mathbb{R}~ \forall i;\, w_i >0~ \forall i;\, \sum_{i=1}^{M} w_i = 1 \right\},$$ is dense in the space of monotonically increasing continuous differentiable functions $f: \mathbb{R} \rightarrow \mathbb{R}$ with $f(t) \rightarrow \infty$ as $t \rightarrow \infty$ and $f(t) \rightarrow -\infty$ as $t \rightarrow -\infty$ with respect to the norm
$$ || f ||_{ {\cal C}^{1}(I) } = \max_{k=0,1} \max_{t \in I} |f^{[k]}(t) | ,$$
on compact intervals $I = [I_0, I_1] \subset \mathbb{R}$.
\end{lemma}
\begin{proof}
Fix any interval $I = [I_0, I_1]$, and sufficiently small $\epsilon >0$. Choose $c_{min} \in (0, \sigma \circ f(I_0))$, and $c_{max} \in (\sigma \circ f(I_1), 1)$. Since $\sigma^{-1}(\cdot) \in {\cal C}^{2}$, there exists $K_1, K_2 >0$ such that
$$ \sup_{y \in [c_{min},c_{max}]} | \nabla\sigma^{-1}(y) | < K_1 < +\infty ,$$
$$ \sup_{y \in [c_{min},c_{max}]} | \nabla^{2} \sigma^{-1}(y) | < K_2 < +\infty .$$
Since $\sigma \circ f(\cdot)$ is ${\cal C}^1$ and monotonic increasing with $\sigma \circ f(t) \rightarrow 1$ as $t \rightarrow \infty$ and $\sigma \circ f(t) \rightarrow 0$ as $t \rightarrow -\infty$, we have by Lemma \ref{lemma_uni}, there exists a function $h(\cdot)$ with the form $h(t) = \sum_{i=1}^{N} w_i \sigma(a_i t + b_i)$, such that
$$ |\sigma \circ f(t) - h(t)| < \epsilon / K ,$$
$$   |(\sigma \circ f)'(t) - h'(t)| < \epsilon / K,$$ 
where $K = \max\{K_1, K_2\}$. Therefore, for $t \in I$ we have
\begin{eqnarray*}
  |f(t) -  \sigma^{-1} \circ h(t) | &=& |\sigma^{-1} \circ \sigma \circ f(t) - \sigma^{-1} \circ h(t) | \\
    &\le& \sup_{y \in [c_{min},c_{max}]} | \nabla\sigma^{-1}(y) | | \sigma \circ f(t) - h(t) | \\
    &\le& K \frac{\epsilon}{K} \\
   &=& \epsilon,
\end{eqnarray*}
where the first inequality follows from the mean value theorem. Similarly,
\begin{eqnarray*}
& & |f'(t) - (\sigma^{-1} \circ h)'(t)| =  | (\sigma^{-1} \circ \sigma \circ f)'(t) - (\sigma^{-1} \circ h)'(t) |  \\
&\le & |(\sigma^{-1} \circ \sigma \circ f)'(t) - (\sigma^{-1})'( \sigma \circ f(t) ) h'(t) | + |(\sigma^{-1})'( \sigma \circ f(t) ) h'(t)  -  (\sigma^{-1} \circ h)'(t)|  \\
& \le & \bigg\{ \sup_{t \in I} |(\sigma^{-1})'( \sigma \circ f(t) )|   \bigg\} | (\sigma \circ f)'(t) - h'(t) | + \bigg\{ \sup_{t \in I} |h'(t)| \bigg\} |(\sigma^{-1})'(\sigma \circ f(t)) - (\sigma^{-1})'(h(t)) |  \\
& \le & C \epsilon 
\end{eqnarray*}
for some constant $C$, since $\sup_{t \in I} |(\sigma^{-1})'( \sigma \circ f(t) )|$ is a constant and $ h'(t) $ is uniformly close to $(\sigma \circ f)'(t)$. The result follows since $\epsilon$ is arbitrary.
\end{proof}
By Lemma \ref{lemma_uni} and \ref{lemma_inv}, for any $\epsilon > 0$ and $k=2,\ldots,d$, we have that for all $x$ in some compact subset of $ \mathbb{R}^{d}$, there exists $\tilde{\theta}_k(\cdot) \in \mathbb{R}^{m_k}$ where $\tilde{\theta}_k(\cdot)$ depends continuously on $x^{(1)}, \ldots, x^{(k-1)}$ such that
$$ |\tilde{T}^{(k)}(x^{(1)}, \ldots, x^{(k)}) - S^{(k)}(x^{(k)}; \tilde{\theta}_k( x^{(1)}, \ldots, x^{(k-1)}))| < \epsilon / 2 ,$$
and
$$ |\nabla_k \tilde{T}^{(k)}(x^{(1)}, \ldots, x^{(k)}) - \nabla S^{(k)}(x^{(k)}; \tilde{\theta}_k( x^{(1)}, \ldots, x^{(k-1)}))| < \epsilon / 2 .$$
To see that $\tilde{\theta}_k(\cdot)$ can be chosen to depend continuously on $(x^{(1)}, \ldots, x^{(k-1)})$, we note that the approximation of $\nabla_k \tilde{T}^{(k)}(x^{(1)}, \ldots, x^{(k)}) $ (considered as a function of $x^{(k)}$) can be obtained through Riemann sums by construction (Lemma 3.1 of \cite{Nestoridis2007}), and by continuity of $\nabla_k \tilde{T}^{(k)}(x^{(1)}, \ldots, x^{(k)})$, the distance $\|\tilde{\theta}_k(x^{(1)}, \ldots, x^{(k-1)})  -\tilde{\theta}_k(y^{(1)}, \ldots, y^{(k-1)})\|$ can be made arbitrarily small if the distance $\|(x^{(1)}, \ldots, x^{(k-1)})' - (y^{(1)}, \ldots, y^{(k-1)})'\|$ is arbitrarily small. This also implies that the dimension of $\tilde{\theta}_k(\cdot)$ must be locally bounded. That is, for any $(x^{(1)}, \ldots, x^{(k-1)})$, there exists some $\delta$ such that the dimension of $\tilde{\theta}_k(y^{(1)}, \ldots, y^{(k-1)})$ is upper bounded if $\|(x^{(1)}, \ldots, x^{(k-1)})' - (y^{(1)}, \ldots, y^{(k-1)})'\| < \delta$. Now, for any compact subset $K \subset \mathbb{R}^{k-1}$, we can construct an open cover $\{U_{\alpha}\}_{\alpha}$ of $K$ where the dimension of $\tilde{\theta}_k(\cdot)$ is upper bounded on each $U_{\alpha}$. Then, there exists a finite subcover $\{U_{\alpha_j}\}_{j=1}^{J}$ of $K$, and therefore the dimension of $\tilde{\theta}_k(\cdot)$ is upper bounded on $K$. We note that requiring $w_i > 0$ does not cause additional difficulty since they can be made arbitrarily small.
\\\\
Now, by the universality of feedforward neural networks with sigmoid activation functions, for any $\delta >0$, we can find a conditional network $\hat{\theta}_k(x^{(1)}, \ldots, x^{(k-1)}; \vartheta_k)$ parameterized by $\vartheta_k$ such that
$$ \|\tilde{\theta}_k(x^{(1)}, \ldots, x^{(k-1)}) - \hat{\theta}_k(x^{(1)}, \ldots, x^{(k-1)}; \vartheta_k)\| < \delta . $$
Since both $S^{(k)}$ and $\nabla S^{(k)}$ have bounded derivatives in any compact inteval, they are uniformly continuous. Therefore, we can choose $\delta$ sufficiently small so that
$$ \|\tilde{\theta}_k(x^{(1)}, \ldots, x^{(k-1)}) - \hat{\theta}_k(x^{(1)}, \ldots, x^{(k-1)}; \vartheta_k)\| < \delta, $$
implies that
$$ |S^{(k)}(x^{(k)}; \tilde{\theta}_k(x^{(1)}, \ldots, x^{(k-1)})) - S^{(k)}(x^{(k)}; \hat{\theta}_k(x^{(1)}, \ldots, x^{(k-1)}; \vartheta_k))| < \epsilon / 2 ,$$
$$ |\nabla S^{(k)}(x^{(k)}; \tilde{\theta}_k(x^{(1)}, \ldots, x^{(k-1)})) - \nabla S^{(k)}(x^{(k)}; \hat{\theta}_k(x^{(1)}, \ldots, x^{(k-1)}; \vartheta_k))| < \epsilon / 2 .$$
By the triangle inequality we then have that
$$  |\tilde{T}^{(k)}(x^{(1)}, \ldots, x^{(k)}) - S^{(k)}(x^{(k)}; \hat{\theta}_k(x^{(1)}, \ldots, x^{(k-1)}; \vartheta_k))| < \epsilon,$$
$$  | \nabla_k \tilde{T}^{(k)}(x^{(1)}, \ldots, x^{(k)}) - \nabla S^{(k)}(x^{(k)}; \hat{\theta}_k(x^{(1)}, \ldots, x^{(k-1)}; \vartheta_k))| < \epsilon ,$$
uniformly for all $x$ a compact subset of $\mathbb{R}^{d}$. We conclude that the triangular maps constructed using neural autoregressive flows are dense in the space of continuous differentiable increasing triangular maps. 
\\\\
In particular, for any $\epsilon >0$ and any compact set $K$, there exists an increasing triangular map $T(\cdot)$ wherein the $k$th component of each map $T^{(k)}(\cdot)$ has the form $(\ref{uni_flow})$, and wherein the corresponding conditional networks are universal approximators (e.g., feedforward neural networks with sigmoid activation functions), such that 
$$ |T^{(k)}(x^{(1)}, \ldots, x^{(k)}) - \tilde{T}^{(k)}(x^{(1)}, \ldots, x^{(k)})| < \epsilon, $$
$$ |\nabla_k T^{(k)}(x^{(1)}, \ldots, x^{(k)}) - \nabla_k \tilde{T}^{(k)}(x^{(1)}, \ldots, x^{(k)})| < \epsilon,$$
for $k=2, \ldots, d$. Using similar arguments we obviously also have that
$$ |T^{(1)}(x^{(1)}) - \tilde{T}^{(1)}(x^{(1)})| < \epsilon, $$
$$ |\nabla_1 T^{(1)}(x^{(1)}) - \nabla_1 \tilde{T}^{(1)}(x^{(1)})| < \epsilon,$$
where $T^{(1)}(x^{(1)}) = S^{(1)}(x^{(1)}; \hat{\theta}_1)$.

Thanks to the triangular structure of the map, these imply that
$$ | \mbox{det} (\nabla T(x)) - \mbox{det} (\nabla \tilde{T}(x)) |,  $$ 
can be made arbitrarily small. By smoothness and boundedness of the target density $\eta(\cdot)$, we have that
$$  | \eta( T(x) ) - \eta( \tilde{T}(x) ) |,  $$
can also be made arbitrarily small. Thus, we have that
\begin{eqnarray*}
& & |\eta(T(x)) \mbox{det}(\nabla T(x)) - \eta( \tilde{T}(x)) \mbox{det}(\nabla \tilde{T}(x)) | \\
& \le&  \eta(T(x)) | \mbox{det}(\nabla T(x)) - \mbox{det}(\nabla \tilde{T}(x)) | + \mbox{det}(\nabla \tilde{T}(x)) | \eta( T(x) ) - \eta( \tilde{T}(x) ) |, 
\end{eqnarray*}
 where the right-hand-side of the above inequality is arbitrarily small as $\epsilon$ is made abitrarily small. This concludes the proof of the theorem.

\bibliographystyle{apalike}
\bibliography{refs}

\begin{thebibliography}{}

\bibitem[Adams et~al., 2009]{Adams2009}
Adams, R.~P., Murray, I., and MacKay, D. J.~C. (2009).
\newblock Tractable nonparametric {B}ayesian inference in {P}oisson processes
  with {G}aussian process intensities.
\newblock In {\em Proceedings of the 26th Annual International Conference on
  Machine Learning}, pages 9--16.

\bibitem[Bachl et~al., 2019]{Bachl_2019}
Bachl, F.~E., Lindgren, F., Borchers, D.~L., and Illian, J.~B. (2019).
\newblock inlabru: an {R} package for {B}ayesian spatial modelling from
  ecological survey data.
\newblock {\em Methods in Ecology and Evolution}, 10(6):760--766.

\bibitem[Barron, 1994]{Barron1994}
Barron, A.~R. (1994).
\newblock Approximation and estimation bounds for artificial neural networks.
\newblock {\em Machine Learning}, 14(1):115--133.

\bibitem[Bogachev et~al., 2005]{Bogachev2005}
Bogachev, V.~I., Kolesnikov, A.~V., and Medvedev, K.~V. (2005).
\newblock Triangular transformations of measures.
\newblock {\em Sbornik: Mathematics}, 196(3):309--335.

\bibitem[Brenier, 1991]{Brenier1991}
Brenier, Y. (1991).
\newblock Polar factorization and monotone rearrangement of vector-valued
  functions.
\newblock {\em Communications on Pure and Applied Mathematics}, 44(4):375--417.

\bibitem[Cybenko, 1989]{Cybenko1989}
Cybenko, G. (1989).
\newblock {Approximation by superpositions of a sigmoidal function}.
\newblock {\em Mathematics of Control, Signals, and Systems}, 2(4):303--314.

\bibitem[Dias et~al., 2008]{Dias2008}
Dias, R., Ferreira, C., and Garcia, N. (2008).
\newblock {Penalized maximum likelihood estimation for a function of the
  intensity of a Poisson point process}.
\newblock {\em Statistical Inference for Stochastic Processes}, 11(1):11--34.

\bibitem[Diggle, 1985]{Diggle1985}
Diggle, P. (1985).
\newblock A kernel method for smoothing point process data.
\newblock {\em Journal of the Royal Statistical Society Series C},
  34(2):138--147.

\bibitem[Dinh et~al., 2015]{Dinh2015}
Dinh, L., Krueger, D., and Bengio, Y. (2015).
\newblock {NICE:} non-linear independent components estimation.
\newblock In {\em Workshop Track Proceedings of the 3rd International
  Conference on Learning Representations}.

\bibitem[Dinh et~al., 2017]{Dinh2017}
Dinh, L., Sohl{-}Dickstein, J., and Bengio, S. (2017).
\newblock Density estimation using real {NVP}.
\newblock In {\em Conference Track Proceedings of the 5th International
  Conference on Learning Representations}.

\bibitem[Efron, 1979]{Efron1979}
Efron, B. (1979).
\newblock Bootstrap methods: Another look at the jackknife.
\newblock {\em Annals of Statistics}, 7(1):1--26.

\bibitem[Efron, 1981]{Efron1981}
Efron, B. (1981).
\newblock {Nonparametric estimates of standard error: The jackknife, the
  bootstrap and other methods}.
\newblock {\em Biometrika}, 68(3):589--599.

\bibitem[Eldan and Shamir, 2016]{Eldan2016}
Eldan, R. and Shamir, O. (2016).
\newblock The power of depth for feedforward neural networks.
\newblock In {\em Proceedings of the 29th Annual Conference on Learning
  Theory}, volume~49 of {\em Proceedings of Machine Learning Research}, pages
  907--940.

\bibitem[Fine et~al., 1999]{Fine1999}
Fine, T.~L., Lauritzen, S.~L., Jordan, M., Lawless, J., and Nair, V. (1999).
\newblock {\em Feedforward Neural Network Methodology}.
\newblock Springer-Verlag, Berlin, Germany, 1st edition.

\bibitem[{Fogel}, 1991]{Fogel1991}
{Fogel}, D.~B. (1991).
\newblock An information criterion for optimal neural network selection.
\newblock {\em IEEE Transactions on Neural Networks}, 2(5):490--497.

\bibitem[Germain et~al., 2015]{Germain2015}
Germain, M., Gregor, K., Murray, I., and Larochelle, H. (2015).
\newblock {MADE}: Masked autoencoder for distribution estimation.
\newblock In {\em Proceedings of the 32nd International Conference on Machine
  Learning}, pages 881--889.

\bibitem[Hong and Guo, 1995]{Hong1995}
Hong, L.-L. and Guo, S.-W. (1995).
\newblock {Nonstationary Poisson model for earthquake occurrences}.
\newblock {\em Bulletin of the Seismological Society of America},
  85(3):814--824.

\bibitem[Hornik et~al., 1989]{Hornik1989}
Hornik, K., Stinchcombe, M., and White, H. (1989).
\newblock Multilayer feedforward networks are universal approximators.
\newblock {\em Neural Networks}, 2(5):359--366.

\bibitem[Huang et~al., 2018]{Huang2018}
Huang, C.-W., Krueger, D., Lacoste, A., and Courville, A. (2018).
\newblock Neural autoregressive flows.
\newblock In {\em Proceedings of the 35th International Conference on Machine
  Learning}, volume~80 of {\em Proceedings of Machine Learning Research}, pages
  2078--2087.

\bibitem[Illian et~al., 2012]{Illian2012}
Illian, J., Sorbye, S., and Rue, H. (2012).
\newblock A toolbox for fitting complex spatial point process models using
  integrated nested {L}aplace approximation ({INLA}).
\newblock {\em Annals of Applied Statistics}, 6(4):1499--1530.

\bibitem[{Jinn-Tsong Tsai} et~al., 2006]{Tsai2006}
{Jinn-Tsong Tsai}, {Jyh-Horng Chou}, and {Tung-Kuan Liu} (2006).
\newblock Tuning the structure and parameters of a neural network by using
  hybrid {T}aguchi-genetic algorithm.
\newblock {\em IEEE Transactions on Neural Networks}, 17(1):69--80.

\bibitem[Kingma et~al., 2016]{Kingma2016}
Kingma, D.~P., Salimans, T., Jozefowicz, R., Chen, X., Sutskever, I., and
  Welling, M. (2016).
\newblock Improved variational inference with inverse autoregressive flow.
\newblock In {\em Advances in Neural Information Processing Systems 29}, pages
  4743--4751.

\bibitem[Kolaczyk, 1999]{Kolaczyk1999}
Kolaczyk, E. (1999).
\newblock Wavelet shrinkage estimation of certain {P}oisson intensity signals
  using corrected thresholds.
\newblock {\em Statistica Sinica}, 9(1):119--135.

\bibitem[Letham et~al., 2016]{Letham2016}
Letham, B., Letham, L.~M., and Rudin, C. (2016).
\newblock Bayesian inference of arrival rate and substitution behavior from
  sales transaction data with stockouts.
\newblock In {\em Proceedings of the 22nd International Conference on Knowledge
  Discovery and Data Mining}, pages 1695--1704.

\bibitem[{Leung} et~al., 2003]{Leung2003}
{Leung}, F. H.~F., {Lam}, H.~K., {Ling}, S.~H., and {Tam}, P. K.~S. (2003).
\newblock Tuning of the structure and parameters of a neural network using an
  improved genetic algorithm.
\newblock {\em IEEE Transactions on Neural Networks}, 14(1):79--88.

\bibitem[Lindgren et~al., 2011]{Lindgren_2011}
Lindgren, F., Rue, H., and Lindstr{\"o}m, J. (2011).
\newblock An explicit link between {G}aussian fields and {G}aussian {M}arkov
  random fields: the stochastic partial differential equation approach.
\newblock {\em Journal of the Royal Statistical Society: Series B},
  73(4):423--498.

\bibitem[Lindqvist, 2006]{Lindqvist2006}
Lindqvist, B.~H. (2006).
\newblock On the statistical modeling and analysis of repairable systems.
\newblock {\em Statistical Science}, 21(4):532--551.

\bibitem[Lloyd et~al., 2015]{Lloyd2015}
Lloyd, C., Gunter, T., Osborne, M.~A., and Roberts, S.~J. (2015).
\newblock Variational inference for {G}aussian process modulated {P}oisson
  processes.
\newblock In {\em Proceedings of the 32nd International Conference on
  International Conference on Machine Learning}, pages 1814--1822.

\bibitem[Marzouk et~al., 2016]{Marzouk2016}
Marzouk, Y., Moselhy, T., Parno, M., and Spantini, A. (2016).
\newblock Sampling via measure transport: An introduction.
\newblock In {\em Handbook of Uncertainty Quantification}, pages 1--41.

\bibitem[McCann, 1995]{Mccann1995}
McCann, R.~J. (1995).
\newblock Existence and uniqueness of monotone measure-preserving maps.
\newblock {\em Duke Mathematical Journal}, 80(2):309--323.

\bibitem[Miranda and Morettin, 2011]{Miranda2011}
Miranda, J. and Morettin, P. (2011).
\newblock Estimation of the intensity of non-homogeneous point processes via
  wavelets.
\newblock {\em Annals of the Institute of Statistical Mathematics},
  63(6):1221--1246.

\bibitem[M{\o}ller et~al., 1998]{Moller1998}
M{\o}ller, J., Syversveen, A.~R., and Waagepetersen, R.~P. (1998).
\newblock Log {G}aussian {C}ox processes.
\newblock {\em Scandinavian Journal of Statistics}, 25(3):451--482.

\bibitem[Nestoridis and Stefanopoulos, 2007]{Nestoridis2007}
Nestoridis, V. and Stefanopoulos, V. (2007).
\newblock Universal series and appoximate identities.
\newblock {\em Technical Report}.

\bibitem[Papamakarios et~al., 2017]{Papamakarios2017}
Papamakarios, G., Pavlakou, T., and Murray, I. (2017).
\newblock Masked autoregressive flow for density estimation.
\newblock In {\em Advances in Neural Information Processing Systems 30}, pages
  2338--2347.

\bibitem[Paszke et~al., 2017]{Paszke2017}
Paszke, A., Gross, S., Chintala, S., Chanan, G., Yang, E., DeVito, Z., Lin, Z.,
  Desmaison, A., Antiga, L., and Lerer, A. (2017).
\newblock Automatic differentiation in {PyTorch}.
\newblock In {\em Advances im Neural Information Processing Systems 30 Workshop
  on Autodiff}.

\bibitem[Raghu et~al., 2017]{Raghu2017}
Raghu, M., Poole, B., Kleinberg, J., Ganguli, S., and Sohl-Dickstein, J.
  (2017).
\newblock On the expressive power of deep neural networks.
\newblock In {\em Proceedings of the 34th International Conference on Machine
  Learning}, volume~70 of {\em Proceedings of Machine Learning Research}, pages
  2847--2854.

\bibitem[Silverman, 1986]{Silverman1986}
Silverman, B.~W. (1986).
\newblock {\em Density Estimation for Statistics and Data Analysis}.
\newblock Chapman \& Hall, London.

\bibitem[Taddy and Kottas, 2010]{Taddy2010}
Taddy, M. and Kottas, A. (2010).
\newblock Mixture modeling for marked {P}oisson processes.
\newblock {\em Bayesian Analysis}, 7(2):335--362.

\bibitem[Villani, 2009]{Villani2009}
Villani, C. (2009).
\newblock {\em Optimal Transport -- Old and New}.
\newblock Springer-Verlag, Berlin, Germany.

\bibitem[Weinan and Wang, 2018]{Weinan2018}
Weinan, E. and Wang, Q. (2018).
\newblock Exponential convergence of the deep neural network approximation for
  analytic functions.
\newblock {\em Science China Mathematics}, 61(10):1733--1740.

\bibitem[{Zammit Mangion} et~al., 2011]{Zammit2011}
{Zammit Mangion}, A., Yuan, K., Kadirkamanathan, V., Niranjan, M., and
  Sanguinetti, G. (2011).
\newblock Online variational inference for state-space models with
  point-process observations.
\newblock {\em Neural Computation}, 23(8):1967--1999.

\bibitem[Zhao and Xie, 1996]{Zhao1996}
Zhao, M. and Xie, M. (1996).
\newblock On maximum likelihood estimation for a general non-homogeneous
  {P}oisson process.
\newblock {\em Scandinavian Journal of Statistics}, 23(4):597--607.

\end{thebibliography}

\end{document}